\documentclass[journal]{IEEEtran}
\IEEEoverridecommandlockouts
\usepackage{placeins}
\usepackage{amsthm}
\usepackage{amsmath}
\usepackage{mathtools}
\usepackage{amssymb}
\usepackage{algorithmic}
\usepackage{algorithm}
\usepackage{verbatim}
\usepackage{cite}
\usepackage{url}
\usepackage{cases}
\usepackage{MnSymbol}
\usepackage[all,cmtip]{xy}
\usepackage{esint}
\usepackage{ulem}
\usepackage{stackrel}
\usepackage{calc}
\usepackage{units}
\newtheorem{lem}{Lemma}
\newtheorem{prop}{Proposition}
\newtheorem{thm}{Theorem}

\newtheorem{definition}{Definition}
\renewcommand{\IEEEQED}{\IEEEQEDopen}
\IEEEoverridecommandlockouts
\begin{document}

\title{CaSCADE: Compressed Carrier and DOA Estimation}

\author{\thanks{
This project has received funding from the European Union's Horizon 2020 research and innovation program under grant agreement No. 646804-ERC-COG-BNYQ, and from the Israel Science Foundation under Grant no. 335/14. Deborah Cohen is grateful to the Azrieli Foundation for the award of an Azrieli Fellowship.}
Shahar Stein, Or Yair, Deborah Cohen, \emph{Student
IEEE} and Yonina C. Eldar, \emph{Fellow IEEE} }

\maketitle

\IEEEpeerreviewmaketitle{}

\begin{abstract}
Spectrum sensing and direction of arrival (DOA) estimation have been thoroughly investigated, both separately and as a joint task. Estimating the support of a set of signals and their DOAs is crucial to many signal processing applications, such as Cognitive Radio (CR). A challenging scenario, faced by CRs, is that of multiband signals, composed of several narrowband transmissions spread over a wide spectrum each with unknown carrier frequencies and DOAs. The Nyquist rate of such signals is high and constitutes a bottleneck both in the analog and digital domains. To alleviate the sampling rate issue, several sub-Nyquist sampling methods, such as multicoset sampling or the modulated wideband converter (MWC), have been proposed in the context of spectrum sensing. In this work, we first suggest an alternative sub-Nyquist sampling and signal reconstruction method to the MWC, based on a uniform linear array (ULA). We then extend our approach to joint spectrum sensing and DOA estimation and propose the CompreSsed CArrier and DOA Estimation (CaSCADE) system, composed of an L-shaped array with two ULAs. In both cases, we derive perfect recovery conditions of the signal parameters (carrier frequencies and DOAs if relevant) and the signal itself and provide two reconstruction algorithms, one based on the ESPRIT method and the second on compressed sensing techniques. Both our joint carriers and DOAs recovery algorithms overcome the well-known pairing issue between the two parameters. Simulations demonstrate that our alternative spectrum sensing system outperforms the MWC in terms of recovery error and design complexity and show joint carrier frequencies and DOAs from our CaSCADE system's sub-Nyquist samples.
\end{abstract}

\section{Introduction}
Both traditional tasks of spectrum sensing and direction of arrival (DOA) estimation have been thoroughly investigated in the literature. For the first, several sensing schemes have been proposed, such as energy detection \cite{Urkowitz_energy}, matched filter \cite{MF1, MF2} and cyclostationary detection \cite{Gardner_review, Napo_review}, assuming known or identical DOAs. Well known techniques for DOA estimation include MUSIC \cite{pisarenko73, schmidt86} and ESPRIT \cite{PAUL1986}. Here, the signal frequency support is typically known. However, many signal processing applications may require or at least benefit from the two combined, namely joint spectrum sensing and DOA estimation.

Cognitive Radio (CR) \cite{Mitola, MitolaMag} is one such application, which aims at solving the spectrum scarcity issue by exploiting its sparsity. Spectral resources, traditionally allocated to licensed or primary users (PUs) by governmental organizations, are becoming critically scant but at the same time have been shown to be underutilized \cite{Study1, Study2}. These observations led to the idea of CR, which allows secondary users to opportunistically access the licensed frequency bands left vacant by their primary owners increasing spectral efficiency \cite{Mitola, Haykin}. Spectrum sensing is an essential task in the CR cycle \cite{cog, MagazineMishali}. Indeed, a CR should be able to constantly monitor the spectrum and detect the PUs' activity, reliably and fast \cite{cognitive1,cognitive2}. DOA recovery can enhance CR performance by allowing exploitation of vacant bands in space in addition to the frequency domain.

The 2D-DOA problem, which requires finding two unknown angles for each transmission and pairing them, is considered in \cite{Jgu07, Jgu15}. The authors suggest a modification to the traditional ESPRIT \cite{PAUL1986}, which is used to estimate a single angle. However, this approach only allows the recovery of two angles and solves a separable problem. This cannot be directly extended to joint angle and frequency estimation, which is not separable. Joint DOA and carrier frequency estimation has been considered in \cite{vanderveen1,vanderveen2}, where the authors developed a joint angle-frequency estimation (JAFE) algorithm. JAFE is based on an extension of ESPRIT which allows for multiple parameters to be recovered. However, this method requires additional joint diagonalization of two matrices using iterative algorithms to pair between the carrier frequencies and the DOAs of the different transmissions. In \cite{bai}, the authors consider multiple interleaved sampling channels, with a fixed delay between consecutive channels. They propose a two-stage reconstruction method, where first the frequencies are recovered and then the DOAs are computed from the corresponding estimated carriers. The works described above all assume that the signal is sampled at least at its Nyquist rate, and do not consider signal reconstruction.

Many modern applications deal with signals with high bandwidth and consequently high Nyquist rate. For instance, to increase the chance of finding unoccupied spectral bands, CRs have to sense a wide spectrum, leading to prohibitively high Nyquist rates. Moreover, such high sampling rates generate a large number of samples to process, affecting speed and power consumption. To overcome the rate bottleneck, several new sampling methods have recently been proposed \cite{Mishali09, Mishali10, MagazineMishali} that reduce the sampling rate in multiband settings below the Nyquist rate.

The multicoset or interleaved approach adopted in \cite{Mishali09} suffers from practical issues, as described in \cite{Mishali10}. Specifically, the signal bandwidth can exceed the analog bandwidth of the low rate analog-to-digital converter (ADC) by orders of magnitude. Another practical issue stems from the time shift elements since it can be difficult to maintain accurate time delays between the ADCs at such high rates. The modulated wideband converter (MWC) \cite{Mishali10} was designed to overcome these issues. It consists of an analog front-end composed of several channels. In each channel, the analog wideband signal is mixed by a periodic function, low-pass filtered and sampled at a low rate.
The MWC solves carrier frequency estimation and spectrum sensing from sub-Nyquist samples, but does not address DOA recovery.

A few works have recently considered joint DOA and spectrum sensing of multiband signals from sub-Nyquist samples. In \cite{leus_DOA}, the authors consider both time and spatial compression by selecting receivers from a uniform linear array (ULA) and samples from the Nyquist grid. They exploit a mathematical relation between sub-Nyquist and Nyquist samples over a certain sensing time and recover the signal's power spectrum from the compressed samples. The frequency support and DOAs are then estimated by identifying peaks of the power spectrum, corresponding to each one of the uncorrelated transmissions. Since the power spectrum is computed over a finite sensing time, the frequency supports and angles are obtained on a grid defined by the number of samples. In \cite{kumar_doa}, an L-shaped array with two interleaved (or multicoset) channels, with a fixed delay between the two, samples the signal below the Nyquist rate. Then, the carrier frequencies and the DOAs are recovered from the samples. However, the pairing issue between the two is not discussed. Moreover, this delay-based approach suffers from the same drawbacks as the multicoset sampling scheme when it comes to practical implementation.

In this work, we first consider spectrum sensing of a multiband signal whose transmissions are assumed to have known or identical DOAs, as in \cite{Mishali10}. For this scenario, we present an alternative sub-Nyquist sampling scheme based on a ULA of sensors. We then extend this scheme to the scenario where both the carrier frequencies and DOAs of the transmissions composing the input signal are unknown. In this case, we propose the CompreSsed CArrier and DOA Estimation (CaSCADE) system, composed of an L-shaped array, and perform joint DOA and carrier recovery from sub-Nyquist samples. 

In the first scenario, we consider a ULA where each sensor implements one channel of the MWC. This configuration has two main advantages over the MWC. First, it allows for a simpler design of the mixing functions which can be identical in all sensors. Second, the ULA based system outperforms the MWC in low signal to noise ratio (SNR) regimes. Since all the MWC channels belong to the same sensor, they are all affected by the same additive sensor noise. In the ULA based system, each channel has a different sensor with uncorrelated sensor noise between channels. This allows for noise averaging which increases the SNR.

We present two approaches to recover the carrier frequencies of the transmissions composing the input signal. The first method is based on compressed sensing (CS) \cite{CSBook} algorithms and assumes that the carriers lie on a predefined grid. In the second technique, we drop the grid assumption and use the ESPRIT algorithm \cite{PAUL1986} to estimate the frequencies. Once these are recovered, we show how the signal itself can be reconstructed. We demonstrate that the minimal number of sensors required for perfect reconstruction in noiseless settings is identical for both recovery approaches and that our system achieves the minimal sampling rate derived in \cite{Mishali09}.

Next, we extend our approach to joint spectrum sensing and DOA estimation from sub-Nyquist samples, using CaSCADE implementing the modified MWC over an L-shape array. Specifically, we consider several narrowband transmissions spread over a wide spectrum, impinging on an L-shaped ULA, each from a different direction. The array sensors are composed of an analog mixing front-end, implementing one channel of the MWC \cite{Mishali10}, as before. We then propose two approaches to jointly recover the carrier frequencies and DOAs of the transmissions. The first is based on CS techniques and allows recovery of both parameters assuming they lie on a predefined grid. The CS problem is formulated in such a way that no pairing issue arises between the carrier frequencies and their corresponding DOAs. 
The second approach, inspired by \cite{Jgu07, Jgu15}, extends the ESPRIT algorithm to the joint estimation of carriers and DOAs, while overcoming the pairing issue. Our 2D-ESPRIT algorithm can be applied to sub-Nyquist samples, as opposed to previous work which only considered the Nyquist regime.

Once the carriers and DOAs are recovered, the signal itself is reconstructed, similarly to the previous scenario. We provide sufficient conditions on our sampling system for perfect reconstruction of the carriers and DOAs, and of the signal itself. We compare our reconstruction algorithms to the Parallel Factor (PARAFAC) analysis method \cite{Hars94}, previously proposed for the 2D-DOA problem \cite{Zhang2011}, \cite{Liu2010}. This approach solves the pairing issue between two estimated angles. However, it has only been applied in the Nyquist regime so far. In \cite{Stein2015}, we applied it on sub-Nyquist samples and extended it to the case where the second variable is a frequency rather than an additional angle. Last, for each scenario, we derive the minimal sampling rate allowing for perfect reconstruction of the signal parameters and the signal itself in noiseless settings.

This paper is organized as follows. In Section~\ref{sec:model}, we formulate the signal model and spectrum sensing goal. Section~\ref{sec:samp_rec} presents the ULA-based sub-Nyquist sampling and reconstruction schemes. Numerical experiments for the spectrum sensing scenario, including comparison with the MWC system, are shown in Section~\ref{sec:exp}. The joint spectrum sensing and DOA estimation problem is considered in Section~\ref{sec:joint}. We present the CaSCADE system along with its sampling scheme and reconstruction techniques, and illustrate its performance in simulations.

\section{Spectrum Sensing Problem Formulation}
\label{sec:model}

\subsection{Signal Model \label{sub:Signal-Model}}

Let $u\left(t\right)$ be a complex-valued continuous-time signal,
bandlimited to $\mathcal{F}=\left[-\frac{f_{\text{Nyq}}}{2},\frac{f_{\text{Nyq}}}{2}\right]$
and composed of up to $M$ uncorrelated transmissions $s_{i}\left(t\right),\,i\in\left\{ 1,2,...,M\right\} $.
Each transmission $s_{i}\left(t\right)$ is modulated by a carrier
frequency $f_{i}\in\mathbb{R}$, such that
\begin{equation}
u\left(t\right) = \sum_{i=1}^{M}s_{i}\left(t\right)e^{j2\pi f_{i}t}.
\end{equation}
Assume that $s_{i}\left(t\right)$ are bandlimited to $\mathcal{B}=\left[-\nicefrac{1}{2T},\nicefrac{1}{2T}\right]$ and disjoint, namely $\min_{i\neq j}\left\{ \left|f_{i}-f_{j}\right|\right\} >B$,
where $B=|\mathcal{B}|$. Formally, the Fourier transform of $u(t)$,
defined by
\begin{equation}
U(f)=\intop_{-\infty}^{\infty}u(t)e^{-j2\pi ft}dt=\sum_{i=1}^{M}S_{i}(f-f_{i}),
\end{equation}
where $S_{i}\left(f\right)$ is the Fourier transform of $s_{i}\left(t\right)$,
is zero for every $f\notin\mathcal{F}$. All source signals are assumed
to have identical and known angle of arrival (AOA) $\theta\neq90^{\circ}$.
A typical source signal $u(t)$ is depicted in the frequency domain
in Fig. \ref{signals figure}(a).

\begin{definition} The set $\mathcal{M}_{1}$ contains all signals
$u\left(t\right)$, such that the support of the Fourier transform
$U\left(f\right)$ is contained within a union of $M$ disjoint intervals
in $\mathcal{F}$. Each of the bandwidths does not exceed $B$ and
all the transmissions composing $u\left(t\right)$ have identical
and known AOA $\theta\neq90^{\circ}$. \end{definition}

We wish to design a sampling and reconstruction system for signals
from the model $\mathcal{M}_{1}$ which satisfies the following properties:
\begin{enumerate}
\item The system has no prior knowledge on the carrier frequencies.
\item The sampling rate should be as low as possible.
\end{enumerate}
Let $\mathbf{s}(t)=\left[s_{1}(t),s_{2}(t),\cdots,s_{M}(t)\right]^{T}$
be the source signals vector, $\mathbf{S}(f)=\left[S_{1}(f),S_{2}(f),\cdots,S_{M}(f)\right]^{T}$
the signal Fourier transform vector, and $\boldsymbol{f}=\left[f_{1},f_{2},\cdots,f_{M}\right]^{T}$
the carrier frequencies vector. Our goal is to design a sampling and
reconstruction system in order to recover $\boldsymbol{f}$ and $\mathbf{s}(t)$
from sub-Nyquist samples of $u(t)$. In the reconstruction phase,
we will address two separate objectives:
\begin{enumerate}
\item Frequencies recovery, i.e. recovering only the signals carrier frequencies
$\boldsymbol{f}$.
\item Full spectrum recovery, i.e. recovering both the signals carrier frequencies
$\boldsymbol{f}$ and the source signals $\mathbf{s}(t)$.
\end{enumerate}

\subsection{Multicoset Sampling and the MWC}

It was previously shown in \cite{Mishali09}, that if $MB<\frac{f_{\text{Nyq}}}{2}$,
then the minimal sampling rate to allow blind reconstruction of $u\left(t\right)$
is $2MB$, namely twice the Landau rate \cite{Landau67}. Concrete
algorithms for blind recovery achieving the minimal rate were developed
in \cite{Mishali09} based on multicoset sampling and in \cite{Mishali10}
based on the MWC. Unfortunately, the implementation of multicoset
sampling is problematic due to the inherent analog bandwidth of
the ADCs and the required synchronization between time shift elements \cite{Mishali10}.

The MWC achieves the minimal sampling rate and can be
implemented in practice \cite{Mishali10}. This system is composed of $N$
parallel channels. Each channel consists of an analog mixing front-end in which $u(t)$ is multiplied by a periodic mixing function $p_{n}(t), 1 \leq n \leq N$. This multiplication aliases the
spectrum, such that each spectral band appears in baseband. We denote by $T_{p}$
the period of $p_n(t)$ and require $f_{p}=1/T_{p}\ge B$. The signal then
goes through a low-pass filter (LPF) with cut-off frequency $f_{s}/2$
and is sampled at rate $f_{s} \geq f_p$. Finally, $u(t)$ is reconstructed
from the low rate samples using CS techniques. An illustration of the MWC is shown in Fig.~\ref{fig:mwc}.

A known difficulty of the MWC is choosing appropriate periodic
functions $p_{n}(t)$ so that their Fourier coefficients fulfill CS requirements. In this work, we suggest an alternative implementation of the MWC, based on a ULA, which overcomes this difficulty, and satisfies the properties described above. Besides, our ULA based system, shown in Fig.~\ref{ULA fig}, is more robust to noise, as we will explain in Section~\ref{sec:exp} and demonstrate via simulations. In Section \ref{sec:joint}, we show how to use this system for DOA recovery.

\begin{figure}
\begin{centering}
\fbox{%
\includegraphics[width = 0.48\textwidth]{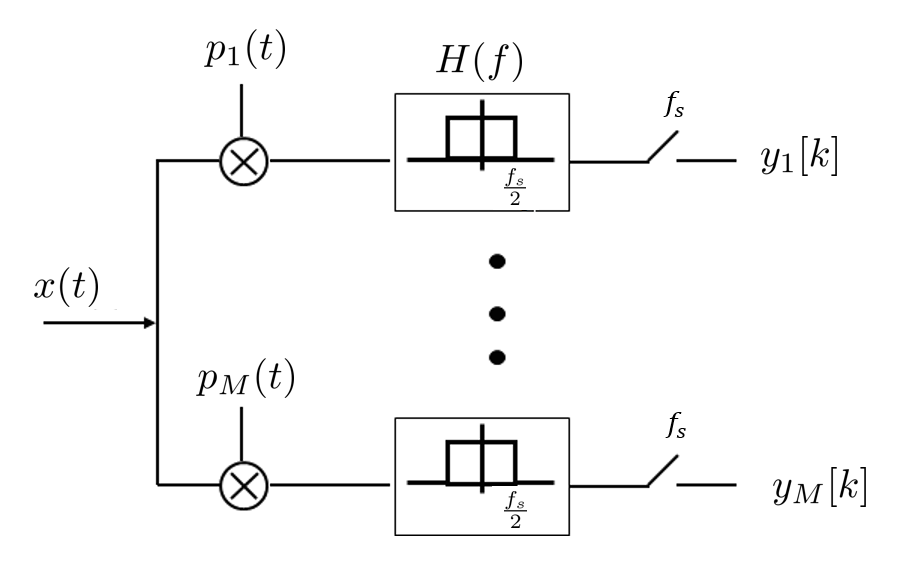}
}
\end{centering}

\protect\caption{MWC system. \label{fig:mwc}}
\end{figure}

\begin{figure*}
\begin{centering}
\fbox{%
\begin{minipage}[t]{1.5\columnwidth}%
\begin{center}
\includegraphics[width = 0.8\textwidth]{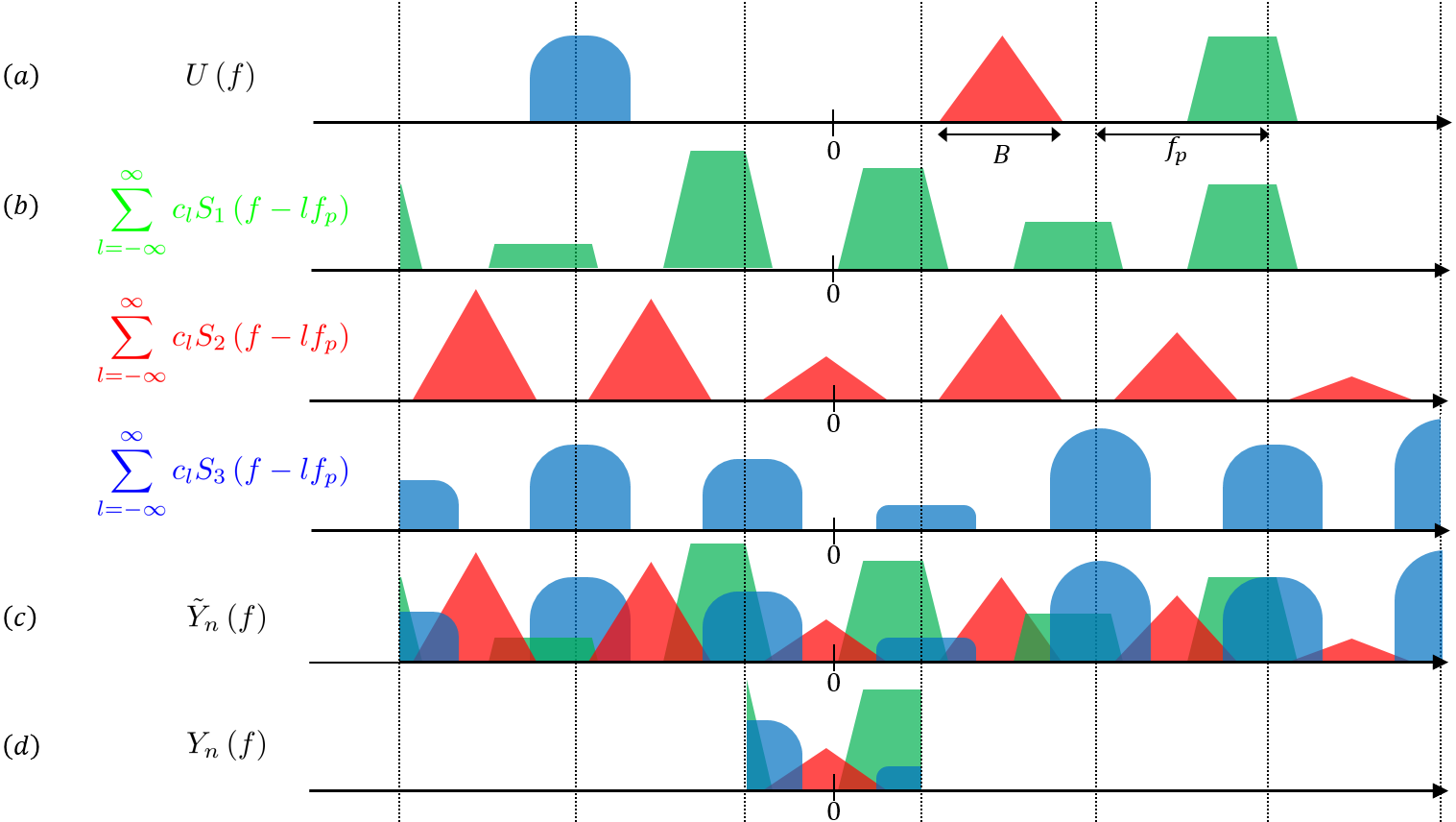}
\par\end{center}%
\end{minipage}}
\par\end{centering}

\protect\caption{The different stages of the analog mixing front-end at the $n$th sensor.
$\left(a\right)$ The input signal in the frequency domain $U\left(f\right)$
with $M=3$ different source signals. $\left(b\right)$ Each replicated
source signal (after mixing). $\left(c\right)$ Replicated input signal
$\tilde{Y}_n\left(f\right)$ (after mixing). $\left(d\right)$ Baseband
signal $Y_n\left(f\right)$ after LPF.\label{signals figure}}
\end{figure*}

\begin{table}
\centering
\protect\caption{Notation}

{\tiny{}{}}%
\begin{tabular}{|c|c|}
\hline
{\tiny{}{}Symbol}  & {\tiny{}{}Interpretation}  \tabularnewline
\hline
\hline
{\tiny{}{}$A$$\left(f\right)$}  & {\tiny{}{}Fourier Transform of $a(t)$}  \tabularnewline
\hline
{\tiny{}{}$A$$\left(e^{j2\pi fT}\right)$}  & {\tiny{}{}DTFT of $a[n]$} \tabularnewline
\hline
{\tiny{}{}$\mathbf{a}$, $\mathbf{A}$}  & {\tiny{}{}vector, matrix (capital letter)} \tabularnewline
\hline
{\tiny{}{}$c$}  & {\tiny{}{}the speed of light}   \tabularnewline
\hline
{\tiny{}{}$\angle\left(\cdot\right)$}  & {\tiny{}{}the angle of $\left(\cdot\right)$, $\angle\left(\cdot\right)\in(-\pi,\pi]$}   \tabularnewline
\hline
$\mathbf{A}^{H}$ & {\tiny{}{}the conjugate-transpose (Hermitian) of $\mathbf{A}$} \tabularnewline
\hline
$\mathbf{A}^{\dagger}$ & {\tiny{}{}the (Moore-Penrose) pseudoinverse of $\mathbf{A}$,
i.e. $\mathbf{A}^{\dagger}=\left(\mathbf{A}^{H}\mathbf{A}\right)^{-1}\mathbf{A}^{H}$}  \tabularnewline
\hline
\end{tabular}
\end{table}

\section{ULA Based MWC}
\label{sec:samp_rec}

\subsection{System Description \label{sec:Sampling-Scheme}}

Our sensing system consists of a ULA composed of $N$ sensors, with two adjacent sensors separated
by a distance $d$, such that $d<\frac{c}{|\cos(\theta)|f_{\text{Nyq}}}$,
where $c$ is the speed of light. All sensors have the same sampling
pattern implementing a single channel of the MWC; the received signal
is multiplied by a periodic function $p(t)$ with period $T_{p}=1/f_{p}$,
low-pass filtered with a filter that has cut-off frequency $f_{s}/2$
and sampled at the low rate $f_{s}$. For simplicity, we choose $f_{s}=f_{p}$.
The system is illustrated in Fig.~\ref{ULA fig}. The only requirement
on $p(t)$ is that none of its Fourier series coefficients within
the signal's Nyquist bandwidth are zero.

In the next section, we show
how we can recover both the carrier frequencies $\boldsymbol{f}$
and $\mathbf{s}(t)$, or alternatively the signal itself,
from the samples at the output of Fig.~\ref{ULA fig}. We demonstrate that the minimal number
of sensors required by both our reconstruction methods is $N=2M$,
with each sensor sampling at the minimal rate of $f_{s}=B$ to allow
for perfect signal recovery. This leads to a minimal sampling rate
of $2MB$, as shown in \cite{Mishali09}, which is assumed to be less than $f_{\text{Nyq}}$. With high probability, the minimal number of sensors reduces to $M+1$.

\begin{figure*}
\centering{}%
\fbox{%
\begin{minipage}[t]{1.5\columnwidth}%
\begin{center}
$$\xymatrix@C=2pc@R=0.7pc{
& \ar@<1ex>[dr]\ar@<-1ex>@{.>}[dr]^<<[@!-41]{s_{1,...,M}}\ar@{-->}[dr] \\
\ar@<1ex>[ddrr]\ar@<-1ex>@{.>}[ddrr]^<[@!+47]{\overset{\overset{\text{wave}}{\text{front}} }{........................}}\ar@{-->}[ddrr] & & *+[o][F]{\sum}\ar[r]\ar@{<.>}[dd]|{d} & \overset{\underset{\downarrow}{p(t)}}{\bigotimes}\ar[r] & *+[F]{\underset{f_s}{LPF}}]\ar[r]_{f_s}|{\underset{}{\nswfilledspoon}} & y_1[k] & \text{ sensor 1} \\
\\
& & *+[o][F]{\sum}\ar[r] & \overset{\underset{\downarrow}{p(t)}}{\bigotimes}\ar[r] & *+[F]{\underset{f_s}{LPF}}]\ar[r]_{f_s}|{\underset{}{\nswfilledspoon}} & y_2[k] & \text{ sensor 2} \\
& \ar@<1ex>[dr]\ar@<-1ex>@{.>}[dr]^<<[@!-43]{s_{1,...,M}}\ar@{-->}[dr] & \vdots & & & \vdots\\
& & *+[o][F]{\sum}\ar[r] & \overset{\underset{\downarrow}{p(t)}}{\bigotimes}\ar[r] & *+[F]{\underset{f_s}{LPF}}]\ar[r]_{f_s}|{\underset{}{\nswfilledspoon}} & y_N[k] & \text{ sensor N} \\
}$$
\par\end{center}%
\end{minipage}}\protect\caption{ULA configuration with $N$ sensors, with distance $d$ between two
adjacent sensors. Each sensor includes an analog front-end composed
of a mixer with the same periodic function $p\left(t\right)$, a LPF
and a sampler, at rate $f_s$.\label{ULA fig} }
\end{figure*}
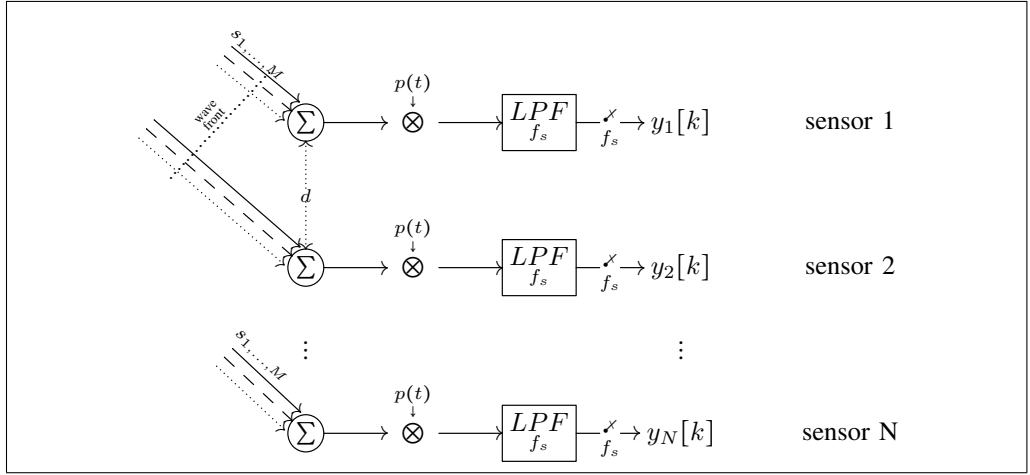

In the remainder of this section, we describe our ULA based sampling scheme and derive
conditions for perfect recovery of the carrier frequencies $\boldsymbol{f}$
and the transmissions $\mathbf{s}(t)$. We then provide concrete recovery
algorithms.

\subsection{Frequency Domain Analysis}
\label{sec:der1}

We start by deriving the relation between the sample sequences from the $n$th sensor and the unknown transmissions $s_{i}(t)$ and
corresponding carrier frequencies $\boldsymbol{f}$. To this end,
we introduce the following definitions
\begin{equation}
\mathcal{F}_{p}\triangleq\left[-\nicefrac{f_{p}}{2},\nicefrac{f_{p}}{2}\right],\quad\mathcal{F}_{s}\triangleq\left[-\nicefrac{f_{s}}{2},\nicefrac{f_{s}}{2}\right].
\end{equation}

Consider the received signal $u_{n}\left(t\right)$ at the $n$th
sensor of the ULA
\begin{equation}
u_{n}(t)=\sum_{i=1}^{M}s_{i}(t+\tau_{n})e^{j2\pi f_{i}(t+\tau_{n})}\approx\sum_{i=1}^{M}s_{i}(t)e^{j2\pi f_{i}(t+\tau_{n})},\label{eq:u_n(t)}
\end{equation}
where 
\begin{equation} \label{eq:delay}
\tau_{n}=\frac{dn}{c}\cos\left(\theta\right)
\end{equation}
is the accumulated
phase at the $n$th sensor with respect to the first sensor. The approximation
in (\ref{eq:u_n(t)}) stems from the narrowband assumption on the
transmissions $s_{i}\left(t\right)$. The Fourier transform of the
received signal $u_{n}\left(t\right)$ is then given by
\begin{equation}
U_{n}\left(f\right)=\sum_{i=1}^{M}S_{i}\left(f-f_{i}\right)e^{j2\pi f_{i}\tau_{n}}.\label{eq:Xi Fourier}
\end{equation}

In each sensor, the received signal is first mixed with the periodic
function $p(t)$ prior to filtering and sampling. Since $p(t)$ is
periodic with period $T_{p}=1/f_{p}$, it can be represented by its
Fourier series
\begin{equation}
p(t)=\sum_{l=-\infty}^{\infty}c_{l}e^{j2\pi lf_{p}t},\label{eq:p(t)}
\end{equation}
where
\begin{equation}
c_{l}=\frac{1}{T_{p}}\intop_{0}^{T_{p}}p(t)e^{-j2\pi lf_{p}t}\mathrm{d}t.\label{eq:c_coeff}
\end{equation}
The Fourier transform of the analog multiplication $\tilde{y}_{n}(t)=u_{n}(t)p(t)$
is evaluated as
\begin{eqnarray}
\tilde{Y}_{n}\left(f\right) & = & \intop_{-\infty}^{\infty}u_{n}\left(t\right)p\left(t\right)e^{-j2\pi ft}\mathrm{d}t\nonumber \\
 & = & \intop_{-\infty}^{\infty}u_{n}\left(t\right)\sum_{l=-\infty}^{\infty}c_{l}e^{j2\pi f_{p}lt}e^{-j2\pi ft}\mathrm{d}t\nonumber \\
 & = & \sum_{l=-\infty}^{\infty}c_{l}\intop_{-\infty}^{\infty}u_{n}\left(t\right)e^{-j2\pi t(f-l\cdot f_{p})}\mathrm{d}t\nonumber \\
 & = & \sum_{l=-\infty}^{\infty}c_{l}U_{n}\left(f-lf_{p}\right).\label{eq:mixed signal}
\end{eqnarray}
The mixed signal $\tilde{Y}_{n}\left(f\right)$ is thus a linear combination
of $f_{p}-$shifted and $c_{l}-$scaled copies of $U_{n}\left(f\right)$.
Since $U\left(f\right)=0,\,\forall f\notin\mathcal{F}$, the sum in
(\ref{eq:mixed signal}) contains at most $\left\lceil \frac{f_{\text{Nyq}}}{f_{p}}\right\rceil $
nonzero terms, for each $f$. Fig. \ref{signals figure}(b)-(c) depicts
each transmission and the resulting signal after mixing, respectively.

Substituting (\ref{eq:Xi Fourier}) into (\ref{eq:mixed signal}), we
have
\begin{eqnarray*}
\tilde{Y}_{n}\left(f\right) & = & \sum_{l=-\infty}^{\infty}c_{l}\sum_{i=1}^{M}S_{i}\left(f-f_{i}-lf_{p}\right)e^{j2\pi f_{i}\tau_{n}}.
\end{eqnarray*}
Denote by $h(t)$ and $H(f)$ the impulse and frequency responses of an ideal LPF with cut-off frequency $f_{s}$, respectively. After
filtering $\tilde{y}_{n}(t)$ with $h(t)$, we have
\begin{eqnarray*}
Y_{n}\left(f\right) & = & \tilde{Y}_{n}\left(f\right) H\left(f\right)\\
 & = & \begin{cases}
\sum_{l=-\infty}^{\infty}c_{l}\sum_{i=1}^{M}S_{i}\left(f-f_{i}-lf_{p}\right)e^{j2\pi f_{i}\tau_{n}}, & f\in\mathcal{F}_{s}\\
0, & f\notin\mathcal{F}_{s}.
\end{cases}
\end{eqnarray*}
Note that $Y_{n}\left(f\right)$ only contains frequencies
in the interval $\mathcal{F}_{s}$, due to the lowpass operation. Therefore, it is composed of a finite number of aliases of $U_{n}\left(f\right)$.
Consequently, we can write

\begin{eqnarray*}
Y_{n}\left(f\right) & = & \sum_{l=-L_{0}}^{L_{0}}c_{l} \sum_{i=1}^{M}S_{i}\left(f-f_{i}-lf_{p}\right)e^{j2\pi f_{i}\tau_{n}}\\
 & = & \sum_{i=1}^{M}e^{j2\pi f_{i}\tau_{n}}\sum_{l=-L_{0}}^{L_{0}}c_{l}S_{i}\left(f-f_{i}-lf_{p}\right) \\
 & = & \sum_{i=1}^{M}\tilde{S}_{i}\left(f\right)e^{j2\pi f_{i}\tau_{n}},
\end{eqnarray*}
where $L_{0}$ is the smallest integer such that the sum contains
all nonzero contributions, i.e. $L_{0}=\left\lceil \frac{f_{\text{Nyq}}}{2f_{p}}\right\rceil $,
and
\begin{equation}
\tilde{S}_{i} (f) \triangleq \sum_{l=-L_{0}}^{L_{0}}c_{l}S_{i} (f-f_{i}-lf_{p}).\label{eq:S_tidle}
\end{equation}
The corresponding $Y_n\left(f\right)$ after filtering is depicted in
Fig.~\ref{signals figure}(d). Note that in the interval $\mathcal{F}_{p}$,
$\tilde{S}_{i}\left(f\right)$ is a cyclic shifted and scaled (by
known factors $\left\{ c_{l}\right\} $) version of $S_{i}\left(f\right)$, as shown in Fig.~\ref{s_F and omega f}.

\begin{figure}
\begin{centering}
\fbox{%
\begin{minipage}[t]{1\columnwidth}%
\begin{center}
\includegraphics[width = 0.5\textwidth]{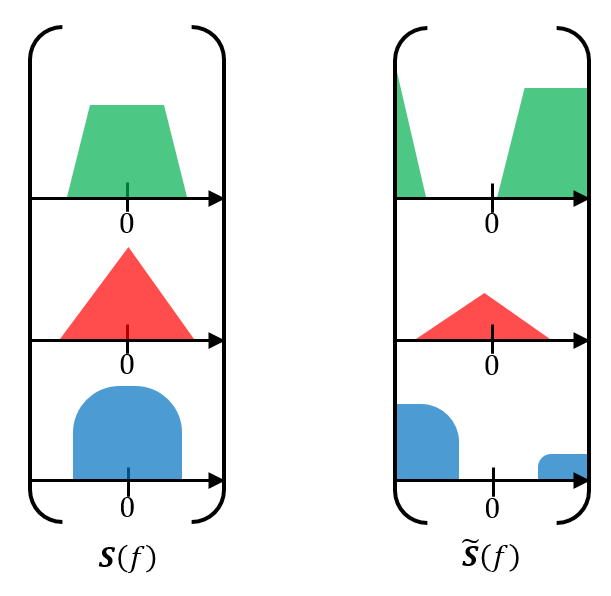}
\par\end{center}%
\end{minipage}}
\par\end{centering}

\protect\caption{The left pane shows the original source signals at baseband (before
modulation). The right pane presents the output signals at baseband
$\tilde{S}\left(f\right)$ after modulation, mixing and filtering.\label{s_F and omega f}}
\end{figure}

After sampling, the discrete-time Fourier transform (DTFT) of the
$n$th sequence $x_{n}\left[k\right]\triangleq y_{n}\left(kT_{s}\right)$
is expressed as
\begin{equation}
X_{n}\left(e^{j2\pi fT_{s}}\right)=\sum_{i=1}^{M}W_{i}\left(e^{j2\pi fT_{s}}\right)e^{j2\pi f_{i}\tau_{n}},\quad f\in\mathcal{F}_{s},\label{eq:dtft_rel}
\end{equation}
where we define $w_{i}\left[k\right]\triangleq\tilde{s}_{i}\left(kT_{s}\right)$
and $W_{i}\left(e^{j2\pi fT_{s}}\right)=\text{DTFT}\left\{ w_{i}\left[k\right]\right\} $.
It is convenient to write (\ref{eq:dtft_rel}) in
matrix form as
\begin{equation}
\label{eq:The Equation}
\mathbf{X}\left(f\right)=\mathbf{A} \mathbf{W}\left(f\right), \quad f \in \mathcal{F}_s.
\end{equation}
Here, $\mathbf{X}\left(f\right)$ is of length $N$ with $n$th
element $X_{n}\left(f\right)=X_{n}\left(e^{j2\pi fT_{s}}\right)$, the unknown vector $\mathbf{W}\left(f\right)$ is of length $M$, with its $i$th entry $W_{i}\left(f\right)=W_{i}\left(e^{j2\pi fT_{s}}\right)$ and the matrix $\mathbf{A}$ depends on the unknown carrier frequencies vector $\boldsymbol{f}$, and is defined by
\begin{equation} \label{eq:A}
\mathbf{A}=
\left(\begin{matrix}e^{j2\pi f_{1}\tau_{1}} & \cdots & e^{j2\pi f_{M}\tau_{1}}\\
\vdots &  & \vdots\\
\\
e^{j2\pi f_{1}\tau_{N}} & \cdots & e^{j2\pi f_{M}\tau_{N}}
\end{matrix}\right).
\end{equation}
In the time domain, we have,
\begin{equation}
\mathbf{x}[k]=\mathbf{A}\mathbf{w}[k],\quad k\in\mathbb{Z},
\end{equation}
where $\mathbf{x}[k]$ has $n$th element $x_{n}[k]$ and $\mathbf{w}[k]$
is a vector of length $M$ with $i$th element $w_{i}[k]$.

In the
next section, we derive sufficient conditions for (\ref{eq:The Equation})
to have a unique solution, namely for perfect recovery of the carrier
frequencies $\boldsymbol{f}$ and the transmissions $\mathbf{s}\left(t\right)$
from the low rate samples $\mathbf{x}\left[k\right]$.

\subsection{Choice of Parameters}

In order to enable perfect blind reconstruction of both the carrier
frequencies $\boldsymbol{f}$ and transmissions $\mathbf{s}\left(t\right)$ in noiseless settings,
we first require (\ref{eq:The Equation}) to have a unique solution.
In addition, we need to ensure that $\mathbf{s}\left(t\right)$ can be uniquely recovered from $\mathbf{w}[k], k \in \mathbb{Z}$.
Theorem \ref{thm: the equation uniquness} presents sufficient conditions
for (\ref{eq:The Equation}) to have a unique solution. Then, Theorem
\ref{prop:signals-recovery:-given} specifies sufficient conditions
for perfect recovery of $\mathbf{s}\left(t\right)$.

\subsubsection{Carrier Frequency Recovery}

We first consider sufficient conditions on the ULA configuration that allow for perfect reconstruction of the carrier frequencies $\boldsymbol{f}$.

\begin{thm}
\label{thm: the equation uniquness} Let $u\left(t\right)$ be an
arbitrary signal in $\mathcal{M}_{1}$ and consider a ULA with
spacing $d<\frac{c}{|\cos\left(\theta\right)| f_{\text{Nyq}}}$ and steering matrix $\mathbf{A}$.
If:
\begin{itemize}
\item (c1) $N>2M-\dim\left(\mbox{span}\left(\mathbf{w}\right)\right)$
\item (c2) $\dim\left(\mbox{span}\left(\mathbf{w}\right)\right)\geq1$,
\end{itemize}
then (\ref{eq:The Equation}) has a unique solution $\left(\boldsymbol{f},\mathbf{w}\right)$. \end{thm}

\begin{IEEEproof}
From the assumption of disjoint transmissions, we have $f_{i}\neq f_{j}$,
for $i\neq j$. Thus, if $d<\frac{c}{|\cos\left(\theta\right)| f_{\text{Nyq}}}$, then it holds that $d\neq\frac{ck}{\left| \cos\left(\theta\right) \right|}\cdot\frac{1}{\left|f_{i}-f_{j}\right|},\,\forall k\in\mathbb{Z},\,\forall i\neq j$, and $e^{j2\pi f_{i} \tau_n} \neq e^{j2\pi f_{j} \tau_n}$ for $1 \leq n \leq N-1$, with $\tau_n$ defined in (\ref{eq:delay}). It follows that $\mathbf{A}$ is a Vandermonde matrix with $M \leq N$, and thus, $\mbox{rank}\left(\mathbf{A}\right)=M$.

Since $d<\frac{c}{|\cos\left(\theta\right)| f_{\text{Nyq}}}$, we have that $2\pi f_{i} \tau_1 \in(-\pi,\pi]$. The proof then follows directly from Proposition 2 in \cite{Kfir2010}.
\begin{prop}
 [Proposition 2, \cite{Kfir2010}] \label{prop:kfir}
If $\left(\boldsymbol{f},\mathbf{w}\right)$ is a
solution to (\ref{eq:The Equation}),
\begin{equation*}
N>2M-\dim\left(\mbox{span}\left(\mathbf{w}\right)\right), \quad 
\dim\left(\mbox{span}\left(\mathbf{w}\right)\right)\geq1
\end{equation*}
then $\left(\boldsymbol{f},\mathbf{w}\right)$ is the unique solution of (\ref{eq:The Equation}).\end{prop} \vspace{-1.75em}
\end{IEEEproof}
Note that $\dim\left(\mbox{span}\left(\mathbf{w}\right)\right) < 1$ iff $u(t) \equiv 0$, that is the received signal does not contain any transmission.


\subsubsection{Signal Recovery}

While Theorem \ref{thm: the equation uniquness} guarantees the uniqueness
of $\left(\boldsymbol{f},\mathbf{w}\right)$, some
additional conditions need to be imposed in order to uniquely derive $\mathbf{s}\left(t\right)$ from $\mathbf{w}$, as $\mathbf{w}$
is a sampled permutation of $\mathbf{s}\left(t\right)$.
Obviously, in order to be able to achieve perfect reconstruction of
$\mathbf{s}\left(t\right)$, the preprocessing of the signal
(i.e mixing with $p(t)$ and filtering with $h\left(t\right)$) should
not cause any loss of information. The following lemma presents conditions
on $p(t)$ and $H(f)$ so that each entry of the processed
signal vector $\tilde{\mathbf{S}}(f)$ is a cyclic shift (up to
scaling by known factors $\{c_{l}\}$) of the matching entry of the
original source signal vector $\mathbf{S}(f)$, as shown in Fig.
\ref{s_F and omega f}. In particular, the transformation between $\mathbf{S}(f)$ and $\tilde{\mathbf{S}}(f)$ should be invertible so that the former can be recovered from the latter. 

\begin{lem}
\label{lem:singal equality} If $f_s \geq f_{p}\geq B$ and $c_{l}\neq0$ for all $l\in\left\{ -L_{0},...,L_{0}\right\} $,
where $c_{l}$ is defined in (\ref{eq:c_coeff}), then
\begin{equation}
\forall f'\in\mathcal{F}_{p},\exists k:\,\tilde{S}_{i}\left(f'\right)=c_{k}S_{i}\left(f'-f_{i}-kf_{p}\right).\label{eq:uni1}
\end{equation}
\end{lem}
\begin{IEEEproof}
Consider the $i$th transmission. The output of the LPF $H(f)$, namely
$\tilde{S}_{i}\left(f\right)$, is given by
\begin{equation}
\tilde{S}_{i}\left(f\right)=\begin{cases}
\sum_{l=-L_{0}}^{L_{0}}c_{l}S_{i}\left(f-f_{i}-lf_{p}\right), & f\in\mathcal{F}_{p}\\
0, & f\notin\mathcal{F}_{p}.
\end{cases}\label{eq:sum_lpf}
\end{equation}
Since $f_{p}\geq B$, the sum in (\ref{eq:sum_lpf}) is over disjoint
bands and only one of its elements is nonzero for each $f$. Equation (\ref{eq:uni1}) is true for $k$ that satisfies $f'-f_{i}-kf_{p}\in\mathcal{F}_{p}$, since for any other $k'\neq k$, $f'-f_{i}-k'f_{p}\notin\mathcal{F}_{p}$
and $\tilde{S}_{i}\left(f'-f_{i}-k'f_{p}\right)=0$.
\end{IEEEproof}


Moreover, if $f_s \geq f_p \geq B$, then the system sampling rate obeys the Nyquist
rate of $\tilde{S}_{i}\left(f\right)$, which means that $\mathbf{S}\left(t\right)$
can be perfectly recovered from $\mathbf{w}\left[k\right]$ and
it holds that
\begin{equation}
W_{i}\left(e^{j2\pi fT_{s}}\right)=\tilde{S}_{i}\left(f\right)\qquad f\in\mathcal{F}_{s}.%
\end{equation}

Theorem \ref{prop:signals-recovery:-given} summarizes sufficient conditions for perfect blind reconstruction
of $\mathbf{s}(t)$ from the low rate samples $\mathbf{x}[k]$.
\begin{thm}
\label{prop:signals-recovery:-given} Let $u(t)$ and the ULA be as
in Theorem \ref{thm: the equation uniquness} and let $\left(\boldsymbol{f},\mathbf{w}\right)$
be the unique solution of (\ref{eq:The Equation}). If:

\label{enu: signals recovery cond1}
\begin{itemize}
\item {(c1)} $c_{l}\neq0$ for all $l\in\left\{ -L_{0},...,L_{0}\right\} $,
where $c_{l}$ is defined in (\ref{eq:c_coeff})
\item {(c2)} $f_{s} \geq f_{p}\geq B$, \label{enu:signals recovery cond2}
\end{itemize}
then $\left\{ \hat{s}_{i}\left(t\right)\right\} _{i=1}^{M}$ can
be uniquely recovered from $\mathbf{x}\left[k\right]$. \end{thm}

\begin{IEEEproof}
Consider the $i$th transmission and let $f'\in\mathcal{B}\subseteq\mathcal{F}_{p}$.
Since $s_{i}(t)$ is bandlimited to $\mathcal{B}$, it holds that
\begin{equation}
W_{i}\left(e^{j2\pi f'T_{s}}\right)=\tilde{S}_{i}\left(f'\right)=c_{l_{a}} S_{i}\left(f'-f_{i}-l_{a}\cdot f_{p}\right),
\end{equation}
where the last equality follows from Lemma \ref{lem:singal equality}.
Since $c_{l_{a}}\neq0$, we have
\begin{equation}
S_{i}\left(f'-f_{i}-l_{a}\cdot f_{p}\right)=\frac{1}{c_{l_{a}}}W_{i}\left(e^{j2\pi f'T_{s}}\right),
\end{equation}
or, after a change of variables,
\begin{equation}
S_{i}\left(f'\right)=\frac{1}{c_{l_{a}}}W_{i}\left(e^{j2\pi\left(f'+f_{i}+l_{a}\cdot f_{p}\right)T_{s}}\right),\label{eq:sVsw}
\end{equation}
where $l_{a}$ is given by
\begin{equation} \label{eq:la}
l_{a}=\left\lfloor \frac{f_{i}+f'+f_p/2}{f_{p}}\right\rfloor,
\end{equation}
completing the proof.
\end{IEEEproof}
Note that $l_a$, defined in (\ref{eq:la}), can only be the index of one of the two $f_p$-bins that may overlap with the $i$th transmission's support. 

\subsubsection{Minimal Sampling Rate}

It was previously proved in \cite{Mishali09} that the minimal sampling
rate for perfect blind reconstruction of a signal of the model {\emph{$\mathcal{M}_{1}$}
is $2MB$. The sampling rate in our ULA based scheme is governed by $B$ and $\dim\left(\mbox{span}\left(\mathbf{w}\right)\right)$, where $1 \leq \dim\left(\mbox{span}\left(\mathbf{w}\right)\right) \leq M$. Therefore, in the worst case, the minimal sampling rate that can be achieved is $2MB$, in accordance with \cite{Mishali09}. With high probability, $\dim\left(\mbox{span}\left(\mathbf{w}\right)\right) = M$ and the minimal rate becomes as low as $\left(M+1\right)B$.

If our sole objective is carrier frequency recovery, then we can further reduce the sampling rate of each channel $f_s$ below $B$. However, in this case, the signal $W_{i}\left(e^{j2\pi fT_{s}}\right)$
is an aliased version of $\tilde{S}_{i}\left(f\right)$. A possible,
though unlikely, consequence of the aliasing is that for some transmission,
the folded versions of $\tilde{S}_{i}\left(f\right)$ cancel each
other and result in $W_{i}\left(e^{j2\pi fT_{s}}\right)\equiv0$.
In such a case, $W_{i}\left(e^{j2\pi fT_{s}}\right)$ and the corresponding
$i$th column of the steering matrix will not appear in (\ref{eq:The Equation}).
Nevertheless, this unlikely scenario will not affect the recovery
of the other signals carrier frequencies. The carrier frequency recovery
is possible for each $s_{i}(t)$ such that $W_{i}\left(e^{j2\pi fT_{s}}\right)\neq0$,
even if $W_{i}\left(e^{j2\pi fT_{s}}\right)$ has suffered from loss
of information due to folding.

\subsection{Reconstruction Methods}
\label{sec:rec1}

In this section, we propose two carrier frequency reconstruction methods that solve (\ref{eq:The Equation}). The first follows from the ESPRIT algorithm \cite{PAUL1986} while the second is based on CS \cite{CSBook}. Once the carriers are estimated, one can recover the transmissions $\mathbf{s}(t)$ by inverting (\ref{eq:The Equation}) and substitute the recovered $W_i(e^{j2 \pi fT_s})$ into (\ref{eq:sVsw}).

\subsubsection{ESPRIT Approach}

\label{sec:prob1_rec}

One practical method
to obtain a solution $\left(\hat{\boldsymbol{f}},\hat{\mathbf{w}}\right)$
is by using the ESPRIT algorithm \cite{PAUL1986} on the measurement
set $\mathbf{x}[k]$, as in \cite{Kfir2010} (Section C.). We can
either assume that the number of source signals $M$ is known or first
estimate it using the minimum description length (MDL)
algorithm \cite{PAUL1986}, for example.

One of the conditions needed to use ESPRIT is that the correlation
matrix $\mathbf{R}_{w} = \sum_{k \in \mathbb{Z}} \mathbf{w}[k]\mathbf{w}^{H}[k] $
is positive definite. From \cite{Kfir2010} (Proposition 3), if $\dim\left(\mbox{span}\left(\mathbf{w}\right)\right)=M$,
then $\mathbf{R}_{w}\succ0$. Therefore, the authors in \cite{Kfir2010}
distinguish between two cases. The first, where $\mathbf{R}_{w}\succ0$,
is referred to as the uncorrelated case. Here, ESPRIT can be directly
applied on $\mathbf{R} = \sum_{k \in \mathbb{Z}} \mathbf{x}[k]\mathbf{x}^{H}[k]$.
The main steps of ESPRIT are summarized in Algorithm \ref{algo:esprit}. In the algorithm description, $\mbox{eig}\left(\mathbf{\Psi}\right)$ is a vector of the eigenvalues
of $\mathbf{\Psi}$ and the correlation matrix $\mathbf{R}$ is estimated as
\begin{equation} \label{eq:Rest}
\mathbf{R}=\sum_{k=1}^{Q}\mathbf{x}[k] \mathbf{x}^H[k],
\end{equation}
where $Q$ is the number of snapshots for the averaging and $\mathbf{x}[k]$
is the vector of samples from the $k$th snapshot.

\begin{algorithm}[H]
\textbf{\uline{Input:}}\textbf{ }
\begin{itemize}
\item $Q$ snapshots of the sensors measurements $\mathbf{x}[k]$ 
\end{itemize}
\textbf{\uline{Output:}}
\begin{itemize}
\item $\hat{\boldsymbol{f}}$ - estimated carriers frequencies
\end{itemize}
\textbf{\uline{Algorithm:}}
\begin{enumerate}
\item Estimate the sample covariance $\mathbf{R}$ from (\ref{eq:Rest})
\item Decompose $\mathbf{R}$ using the singular value
decomposition: $\mathbf{U},\mathbf{S},\mathbf{V}=\mbox{svd}\left(\mathbf{R}\right)$
\item Extract signal subspace: $\mathbf{U}_{s}=\left[\mathbf{U}^{1},...,\mathbf{U}^{M}\right]$
\item Define: $\mathbf{U}_{1}=\left[\mathbf{U}^{1},...,\mathbf{U}^{M-1}\right]$,
$\mathbf{U}_{2}=\left[\mathbf{U}^{2},...,\mathbf{U}^{M}\right]$
\item Least squares recovery:

\begin{enumerate}
\item $\mathbf{\Psi}=\mathbf{U}_{2}\mathbf{U}_{1}^{\dagger}$
\item $\boldsymbol{f}=\arccos\left[\angle\left(\mbox{eig}\left(\mathbf{\Psi}\right)\right)\right]\cdot\frac{c}{2\pi d}$
\end{enumerate}
\end{enumerate}
\protect\caption{ESPRIT}

\label{algo:esprit}
\end{algorithm}

If $\dim\left(\mbox{span}\left(\mathbf{w}\right)\right)<M$, then
the rank of the correlation matrix $\mathbf{R}$ is less than $M$.
Here, an additional step is implemented to construct a smoothed correlation matrix of rank $M$, before applying ESPRIT.
This case is referred to as the correlated case \cite{Kfir2010}.
The smoothed correlation matrix is given by
\begin{equation}
\bar{\mathbf{R}}=\frac{1}{V}\sum_{l=1}^{V} \sum_{k \in \mathbb{Z}} \mathbf{x}_{l}\left[k\right] \mathbf{x}_{l}^{H}\left[k\right],\label{eq:smoothedR}
\end{equation}
where $V\triangleq N-M$ and
\begin{equation}
\mathbf{x}_{l}\left[k\right]\triangleq\left[\begin{matrix}x_{l}\left[k\right] & x_{l+1}\left[k\right] & \cdots & x_{l+M}\left[k\right]\end{matrix}\right]^{T},\quad1\leq l\leq V.
\end{equation}

Note that in order to be able to construct the smoothed correlation
matrix, one should require $N>2M-\dim\left(\mbox{span}\left(\mathbf{w}\right)\right)$,
which is exactly condition ({c2}) in Theorem \ref{thm: the equation uniquness}.

Once the carrier frequencies $f_{i}$ are recovered, the steering
matrix $\mathbf{A}$, defined in (\ref{eq:A})
can be constructed. The vector $\mathbf{W}(f)$ is then obtained
by inverting the steering matrix,
\begin{equation}
\mathbf{W}(f)=\mathbf{A}^{\dagger}\mathbf{X}(f),\label{eq:rec_sig}
\end{equation}
and the source signal vector is computed using (\ref{eq:sVsw}).

\subsubsection{CS Approach}

Suppose that the carrier frequencies $f_{i}$ lie on a grid $\{\delta l\}_{l=-L}^{L}$,
with $L=\frac{f_{\text{Nyq}}}{2\delta}$. Here, $\delta$ is a parameter
of the recovery algorithm that defines the grid resolution. Equation
(\ref{eq:The Equation}) then becomes
\begin{equation}
\mathbf{x}[k]=\mathbf{G}\mathbf{w}[k],\quad k\in\mathbb{Z},\label{eq:ongrid}
\end{equation}
where ${\bf G}$ is a $N\times(2L+1)$ matrix with $(n,l)$ element
$G_{nl}=e^{j2\pi\tau_{n}l\delta}$. The nonzero elements of the sparse
$(2L+1)\times1$ vector $\mathbf{w}[k]$ have unknown indices $l_{i}=\frac{f_{i}}{\delta}$
for $1\leq i\leq M$.

The set of equations (\ref{eq:ongrid}) represents an infinite number
of linear systems with joint sparsity. Such systems are known as infinite measurement
vectors (IMV) in the CS literature \cite{rembo}. We use the support
recovery paradigm from \cite{Mishali09} that produces a
finite system of equations, called multiple measurement vectors (MMV)
from an infinite number of linear systems. This reduction is performed
by what is referred to as the continuous to finite (CTF) block \cite{rembo,CSBook}.

From (\ref{eq:ongrid}), we have
\begin{equation}
\mathbf{R=GR}_w^g \mathbf{G}^{H}
\end{equation}
where $\mathbf{R}=\sum_{k \in \mathbb{Z}}\mathbf{x}[k]\mathbf{x}^{H}[k]=\int_{f\in\mathcal{F}_{s}}\mathbf{X}(f)\mathbf{X}^{H}(f)\mathrm{d}f$
is a $N\times N$ matrix and $\mathbf{R}_w^g=\sum_{k \in \mathbb{Z}}\mathbf{w}[k]\mathbf{w}^{H}[k]=\int_{f\in\mathcal{F}_{s}}\mathbf{W}(f)\mathbf{W}^{H}(f)\mathrm{d}f$
is a $M\times M$ matrix. We then construct a frame ${\bf V}$ such
that $\mathbf{R=VV}^{H}$. Clearly, there are many possible ways to
select ${\bf V}$. We construct it by performing an eigendecomposition
of ${\bf R}$ and choosing ${\bf V}$ as the matrix of eigenvectors
corresponding to the nonzero eigenvalues. We can then define the following
linear system
\begin{equation}
{\bf V=GU.}\label{eq:CTF}
\end{equation}
From \cite{Mishali09} (Propositions 2-3), the support of
the unique sparsest solution of (\ref{eq:CTF}) is the same as the
support of the original set of equations (\ref{eq:ongrid}). Equation
(\ref{eq:CTF}) can be solved using any MMV CS algorithm, such as
simultaneous orthogonal matching pursuit (SOMP) \cite{CSBook}.

Once the support $S$ of ${\bf U}$, namely the support of $\mathbf{W}(f)$,
is recovered, the carrier frequencies $f_{i}$ are computed using
$f_{i}=l_{i}\delta$, with $l_{i}\in S$, and the steering matrix
$\mathbf{A}$, defined in (\ref{eq:A}) is constructed. The vectors ${\mathbf{w}[k]}$ and $\mathbf{s}(t)$ are then obtained using (\ref{eq:rec_sig}) and (\ref{eq:sVsw}), respectively.

Theorem \ref{thm:cs} shows that the conditions for perfect recovery of $\mathbf{w}[k]$ from (\ref{eq:ongrid}) are identical to those derived in the previous section.

\begin{thm}
\label{thm:cs} Let $u(t)$ be an arbitrary signal within $\mathcal{M}_{1}$
and consider a ULA with spacing $d<\frac{c}{|\cos(\theta) |f_{\text{Nyq}}}$.
The minimal number of sensors required for perfect recovery of $\mathbf{w}[k]$
in (\ref{eq:ongrid}) in a noiseless environment is $N> 2M-\text{dim}(\text{span}(\mathbf{w}))$.
\end{thm}
Theorem \ref{thm:cs} follows directly from the fact that if $d<\frac{c}{|\cos(\theta)| f_{\text{Nyq}}}$,
then ${\bf G}$ is a Vandermonde matrix, and therefore has full spark, namely $\text{spark}(\mathbf{G})=M$. Then, we use the MMV recovery condition from \cite{rankCS}, given by
\begin{equation}
M < \frac{\text{spark}(\mathbf{G}) -1 + \text{rank}(\mathbf{V})}{2},
\end{equation}
where $1 \leq \text{rank}(\mathbf{V}) \leq M$. Finally, it holds that
\begin{equation}
\text{rank}(\mathbf{V}) = \text{dim}(\text{span}(\mathbf{x})) = \text{dim}(\text{span}(\mathbf{w})),
\end{equation}
where the last equality follows from the fact that $\bf G$ is full spark.

In the worst case, it holds that $\text{rank}(\mathbf{V}) = \text{dim}(\text{span}(\mathbf{w})) = 1$ and the MMV processing does not improve the recovery ability over the single measurement vector (SMV) case. The required number of sensors is then $2M$, leading to a minimal sampling rate of $2MB$. With high probability, $\text{rank}(\mathbf{V})=\text{dim}(\text{span}(\mathbf{w}))=M$ and the number of sensors required is thus reduced to $M+1$.

\subsection{Comparison with the MWC}

Both our ULA based system and the MWC \cite{Mishali10} allow for reconstruction of multiband signals from samples obtained below the Nyquist rate. The ULA approach adopts the same sampling principle as the MWC but differs in some essential ways. First, the MWC uses one sensor composed of $N$ analog processing channels, whereas the
ULA scheme uses $N$ sensors, each composed of one channel.
While both systems use the same amount of $N$ mixers, LPFs and samplers,
this difference of configuration leads to essential distinctions between
the systems. Since all the MWC channels belong to the same sensor,
they are all affected by the same additive sensor noise, i.e. $\tilde{u}(t)=u(t)+\eta(t)$
in all channels. In the ULA, each channel belongs to a different sensor and as a consequence, is corrupted by a different additive sensor
noise, namely $\tilde{u}_{n}(t)=u_{n}(t)+\eta_{n}(t)$, where $\eta_{n}(t)$
can be assumed to be uncorrelated between channels. This is an advantage
of the ULA based method since the noise is averaged. Moreover, a known
difficulty of the MWC is choosing appropriate mixing functions $\left\{ p_{n}(t)\right\} $
so that the original signal can be reconstructed. The ULA scheme allows for all sensors to use the same function $p\left(t\right)$,
and this function does not have any limitation other than $f_{p}>B$
and $c_{l}\neq0$ for all $l\in\left\{ -L_{0},...,L_{0}\right\} $. Finally, the ULA configuration can be extended to allow for joint carrier and DOA recovery, as shown in Section \ref{sec:joint}.
Table \ref{table:comp} summarizes the main properties of each system.

\begin{table}
{\tiny{}{}}{\tiny \par}
\centering

\begin{tabular}{|c|c|c|}
\hline
 & {\tiny{}{}ULA based system}  & {\tiny{}{}MWC}\tabularnewline
\hline
\hline
{\tiny{}{}Periodic functions}  & {\tiny{}{}one function for all sensors}  & {\tiny{}{}one function per channel}\tabularnewline
\hline
{\tiny{}{}Number of samplers}  & {\tiny{}{}$N$ - number of sensors}  & {\tiny{}{}$N$ - number of channels}\tabularnewline
\hline
{\tiny{}{}Minimal sampling rate (average)}  & {\tiny{}{}$\left(M+1\right)B$}  & {\tiny{}{}$(M+1)B$}\tabularnewline
\hline
{\tiny{}{}Minimal sampling rate (worst)}  & {\tiny{}{}$2MB$}  & {\tiny{}{}$2MB$}\tabularnewline
\hline
{\tiny{}{}Practical sampling rate}  & {\tiny{}{}$Nf_{s}$}  & {\tiny{}{}$Nf_{s}$}\tabularnewline
\hline
\end{tabular}

\protect

\

\caption{Main properties of the ULA based and MWC system.}

\label{table:comp}
\end{table}

\section{Numerical Experiments}
\label{sec:exp}

\label{sec:sim}

We now numerically investigate different aspects of our
system shown in Fig.~{\ref{ULA fig}} and show that it ourperforms the MWC system \cite{Mishali10} of Fig.~\ref{fig:mwc} at low SNRs in terms of recovery error. We first explore the impact of SNR, sampling rate, sensors distance $d$, number of sensors/channels $N$ and number of snapshots $Q$ on the signal reconstruction performance. We then consider carrier frequency recovery only and demonstrate that the sampling rate can be made lower than the Landau rate \cite{Landau67} in this case. For the
ULA based system, we show both the MMV CS and ESPRIT approaches described in Section \ref{sec:rec1}. For the MMV method, we use SOMP \cite{CSBook} for recovering the support $S$.

\subsection{Simulation Setup}

The setup described hereafter is used as a basis for all simulations.
Consider signals of the model $\mathcal{M}_{1}$ with $M=3$, $f_{\text{Nyq}}=10$GHz, $\theta=0^{\circ}$ and $B=50$MHz. The carrier frequencies $f_{i}$ are drawn uniformly at random from $[-\frac{f_{\text{Nyq}}-B}{2},\frac{f_{\text{Nyq}}-B}{2}]$.
In our ULA based system, the received signal at each sensor is given
by (\ref{eq:u_n(t)}). In each sensor, the received signal is corrupted
by uncorrelated additive white Gaussian noise (AWGN) $\eta_{n}(t)$,
such that the signal at the $n$th sensor is given by $\tilde{u}_{n}(t)=u_{n}(t)+\eta_{n}(t)$.
For the MWC system, the received signal is the sum of the transmissions with AWGN, namely $u(t)=\sum_{i=1}^{M}s_{i}(t)e^{j2\pi f_{i}t}+ \eta(t)$. Here, all channels are corrupted by the same noise $\eta(t)$ since they all belong to one unique sensor. The noises $\eta(t)$ and $\eta_n(t), 0 \leq n \leq N-1$ are assumed to have the same variance.

In all the simulations,
we use $f_{s}=f_{p}=1.3B$ (if not mentioned otherwise). For the ULA
based system, we use a periodic function $p\left(t\right)$ such that
$P\left(f\right)=\sum_{l=-\infty}^{\infty}\delta\left(f-lf_{p}\right)$.
In the MWC, $p_{i}(t)$ are chosen as piecewise constant functions
alternating between the levels $\pm1$ with sequences generated uniformly
at random. The system performance is measured by computing the MSE
between the original and reconstructed signals, i.e. $\text{MSE}=||u-\hat{u}||^{2}$ normalized to the length of $u$.
For the simulations, we estimate $\mathbf{R}$ as in (\ref{eq:Rest}).

To set similar conditions for both systems (MWC and ULA based), we
use the same parameters, i.e the number of source signals $M$, the number of snapshots $Q$, SNR, $f_{p}=f_{s}$, and $N$, which in the ULA denotes the sensors number and in the MWC denotes the channels
number. The same signal is fed to both systems at each realization
of the simulations. The results are averaged over $2000$ realizations, where in each realization, the signals, carriers and noises are generated at random.

\subsection{Signal Reconstruction}

Figure \ref{d simulation} presents the performance of
the ULA based system as a function of $d$. As shown in Theorem \ref{thm: the equation uniquness},
we require $d\leq\frac{c}{|\cos\left(\theta\right)|\cdot f_{\text{Nyq}}}$,
which in our setting translates to $d\leq\frac{3\cdot10^{8}}{10^{10}}=0.03[m]$.
This property of the system geometry is clearly demonstrated in Fig.
\ref{d simulation}, where we observe a monotonic decrease in the
performance starting from $d=0.03$, for both reconstruction methods,
MMV and ESPRIT. The decrease in performance below $d=0.03$
stems from the fact that the closer the sensors, the more correlated their samples. In the following simulations, $d$ is set to $d=0.03$.

\begin{figure}
\includegraphics[width = 0.5\textwidth]{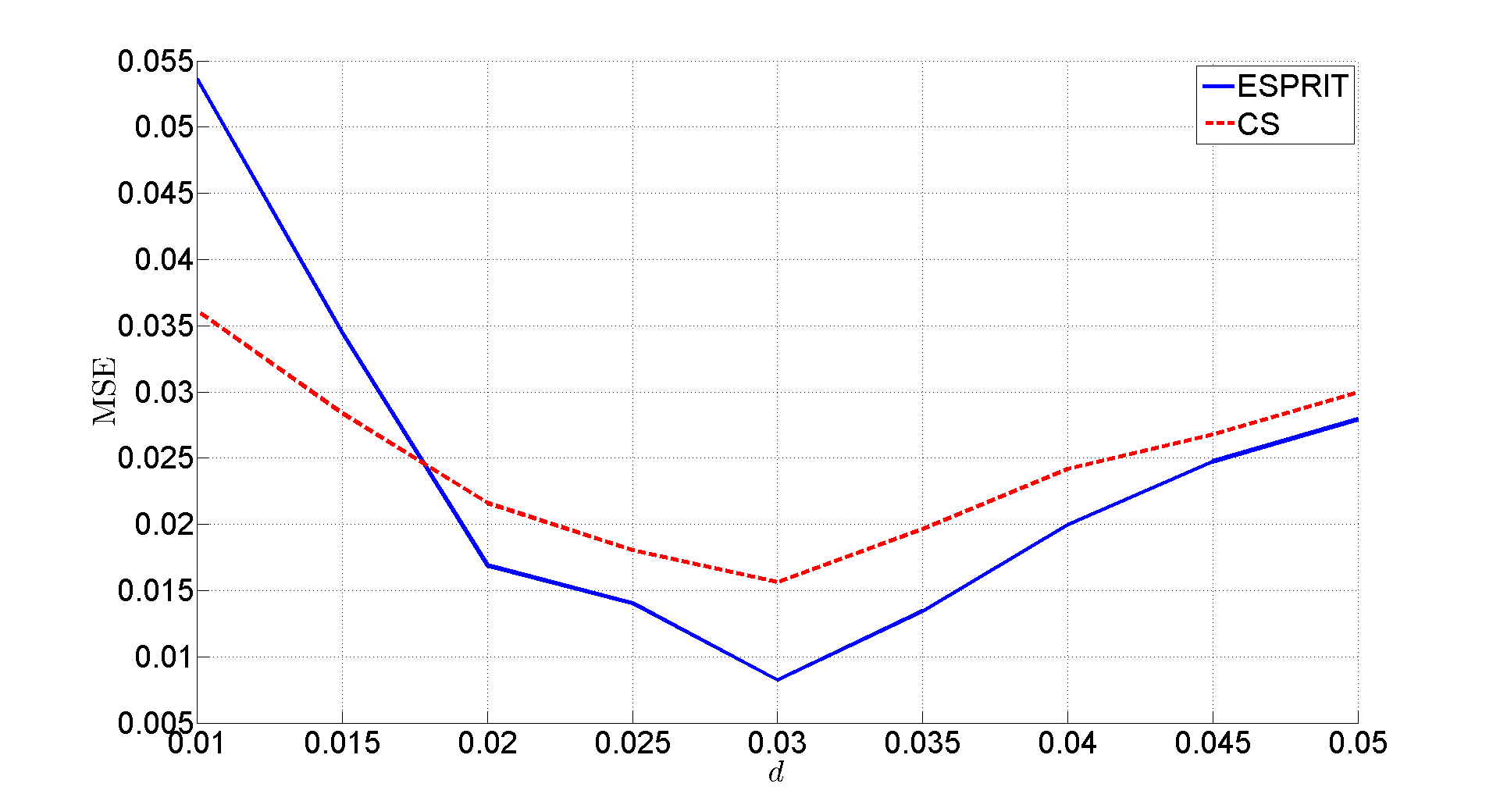}

\protect\caption{Influence of the distance $d$ between adjacent sensors with $M=3$, $N=10$, $Q=400$,
and $\text{SNR}=10$dB.\label{d simulation}}
\end{figure}

We next examine the effect of $f_{p}$. From Theorem \ref{prop:signals-recovery:-given}, $f_{p}$
must be greater than the transmissions bandwidth $B$. When $f_{s}=f_{p}<B$, mixing the signal $u\left(t\right)$ with $p\left(t\right)$ results in aliasing of $u\left(t\right)$, as adjacent shifted copies
of the source signal overlap. Each spectral bin overlaps
with two others over a bandwidth $b=B-f_{p}$ each.
It follows that we reconstruct the aliased version of each signal, that is only $B-2b$ of each source signal's support is perfectly recovered, while the remaining $2b$ are corrupted. Therefore, the reconstruction performance depends on
$\frac{f_{p}}{B}$. In particular, if $\frac{f_{p}}{B}\leq\frac{1}{2}$,
no reconstruction at all is possible. This phenomenon is demonstrated in Fig. \ref{Fp simulation}.

\begin{figure}
\includegraphics[width = 0.5\textwidth]{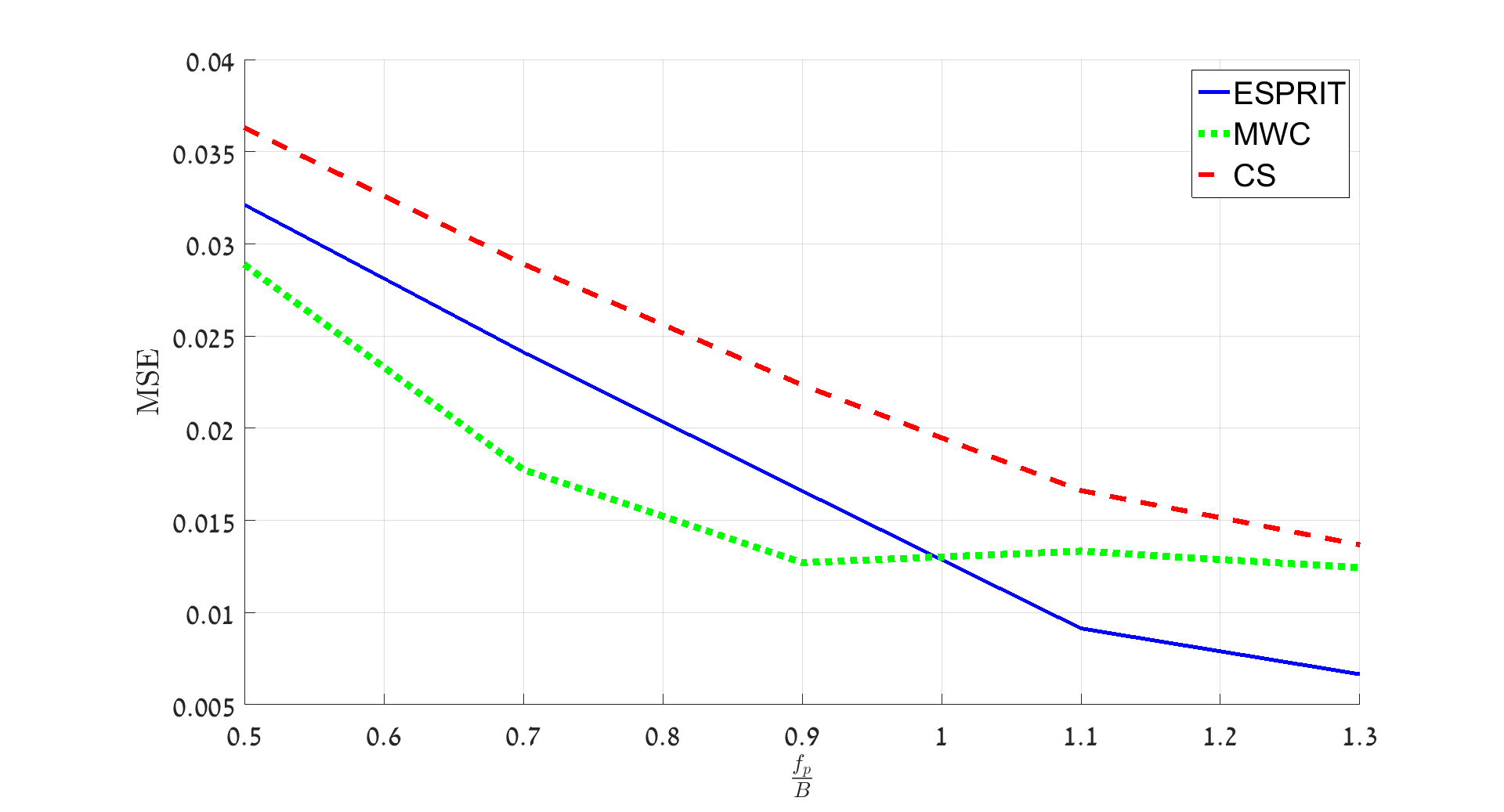}

\protect\caption{Influence of the ratio $\nicefrac{f_{p}}{B}$, with $M=3$, $N=10$, $Q=400$,
and $\text{SNR}=10$dB. \label{Fp simulation}}
\end{figure}

The third experiment examines the influence of the number of sensors
$N$. A large amount of sensors increases the system's robustness to noise and allows it to handle a greater amount of source signals. This parameter
is equivalent to the number of channels in the MWC system. From Fig. \ref{N simulation}, it can be seen that the reconstruction
error decreases with more sensors. In this setting, the minimal number of sensors is $N=2M=6$.

\begin{figure}
\includegraphics[width = 0.5\textwidth]{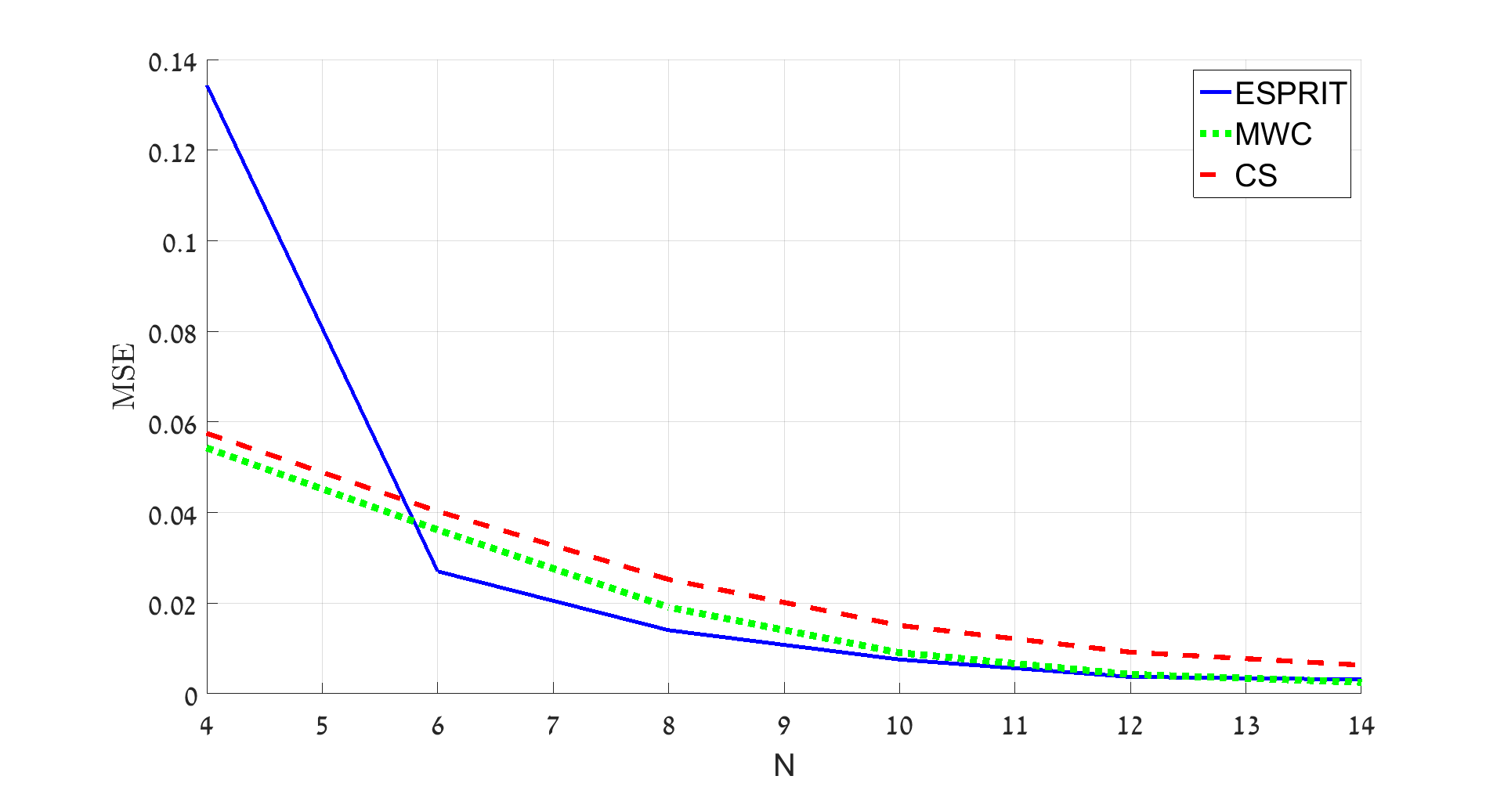}

\protect\caption{Influence of the number of sensors $N$, with $M=3$, $Q=400$,
$\text{SNR}=10$dB. \label{N simulation}}
\end{figure}

The influence of the number of snapshots $Q$ is investigated in the
next experiment. As shown in Fig. \ref{Q simulation}, the performance of ESPRIT improves with the number of snapshots. A small amount of snapshots can yield $\dim\left(\mbox{span}\left(\mathbf{w}\right)\right)<M$.
In this case, referred to as the correlated case, we need to construct
a smoothed correlation matrix (\ref{eq:smoothedR}) on which we can
apply ESPRIT, as discussed in Section \ref{sec:rec1}. This setting can be useful when only carrier frequencies
recovery is needed, as it enables good frequency recovery with few
samples. In this experiment, we set $M=8$ and use a low number
of snapshots. In Fig. \ref{few Q simulation}, we observe that for $Q\leq M=8$
the smoothing algorithm yields better performance than the traditional ESPRIT.
\begin{figure}
\includegraphics[width = 0.5\textwidth]{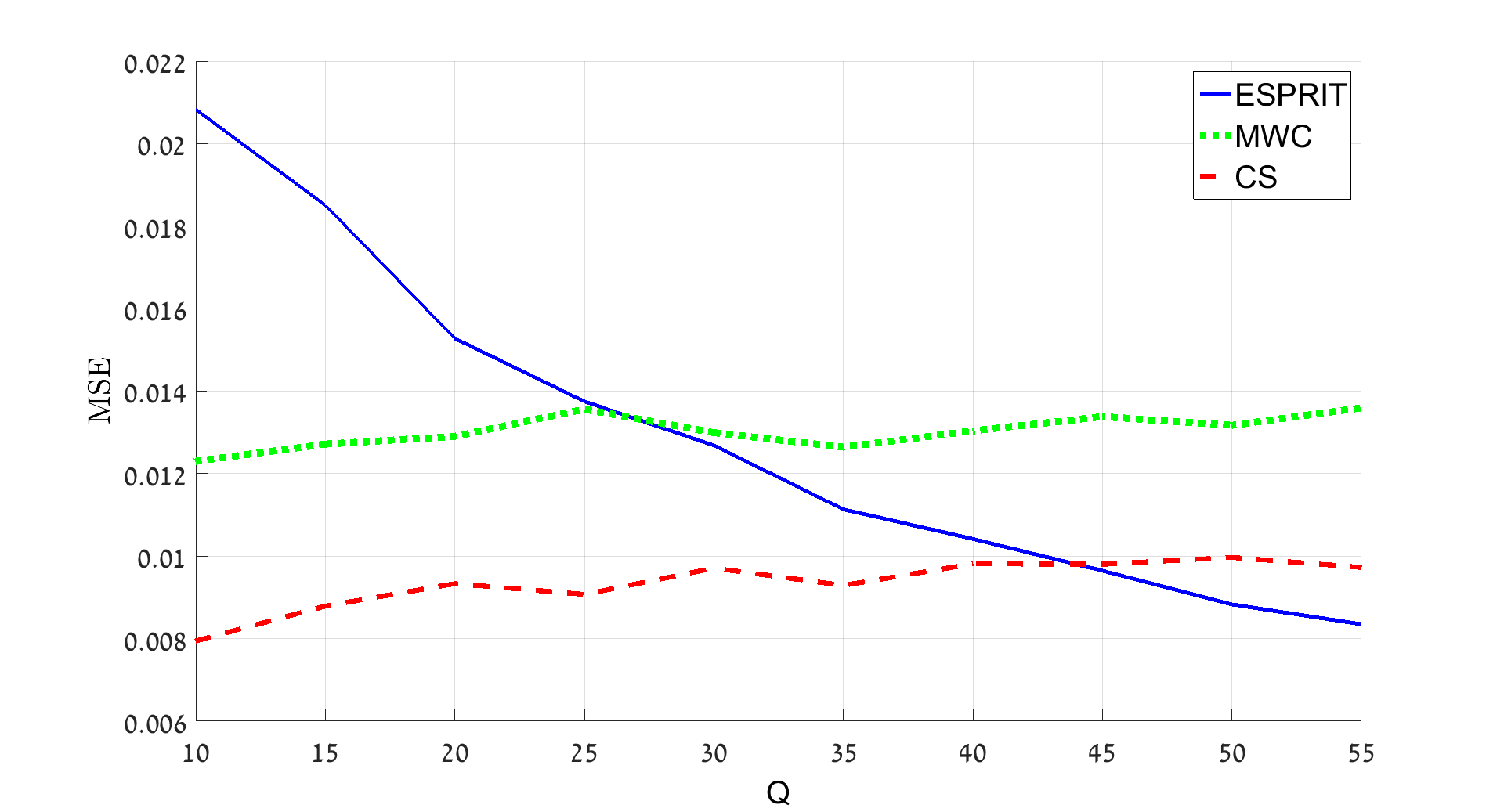}

\protect\caption{Signal reconstruction performance vs. $Q$, with $M=3$, $\text{SNR}=10$dB,
$N=8$. \label{Q simulation}}
\end{figure}

\begin{figure}
\includegraphics[width = 0.5\textwidth]{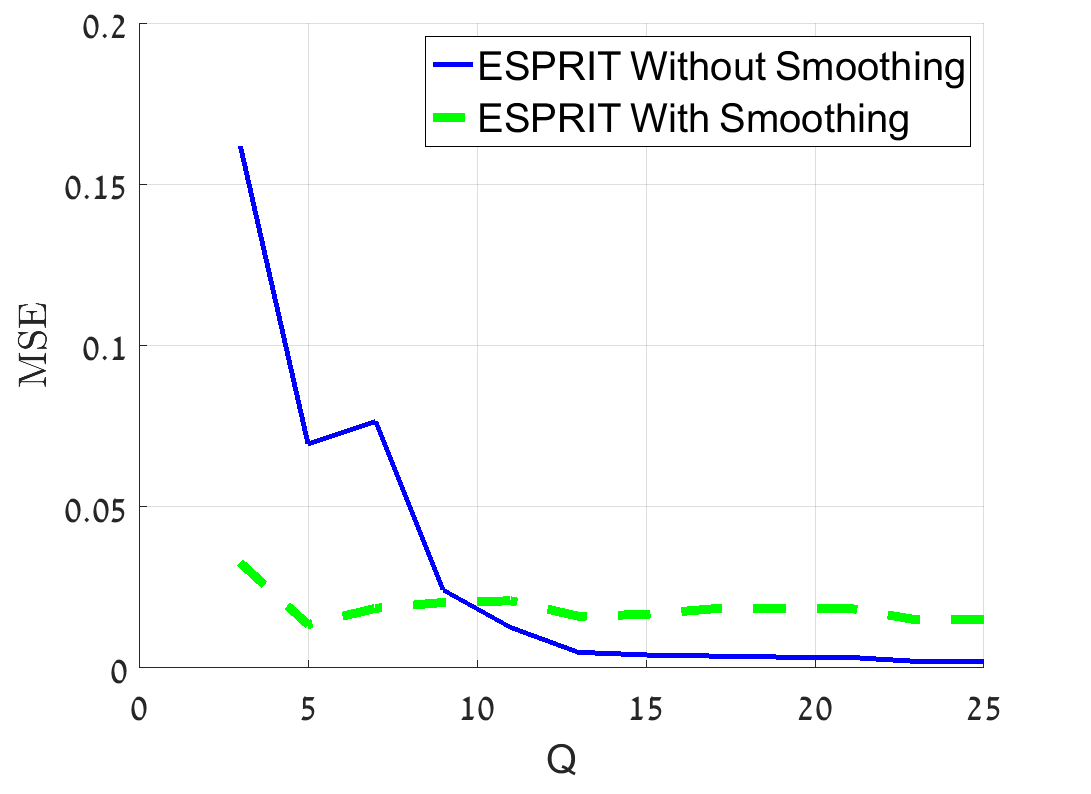}

\protect\caption{Correlated vs. uncorrelated case with $M=8$, $N=25$, $\text{SNR}=20$dB. \label{few Q simulation}}
\end{figure}

The next simulation tests the reconstruction performance under different
SNR conditions. When dealing with low SNR scenarios, grid search
algorithms (such as MMV) are known to outperform analytic algorithms
such as ESPRIT, as illustrated in Fig. \ref{SNR simulation}.

\begin{figure}
\includegraphics[width = 0.5\textwidth]{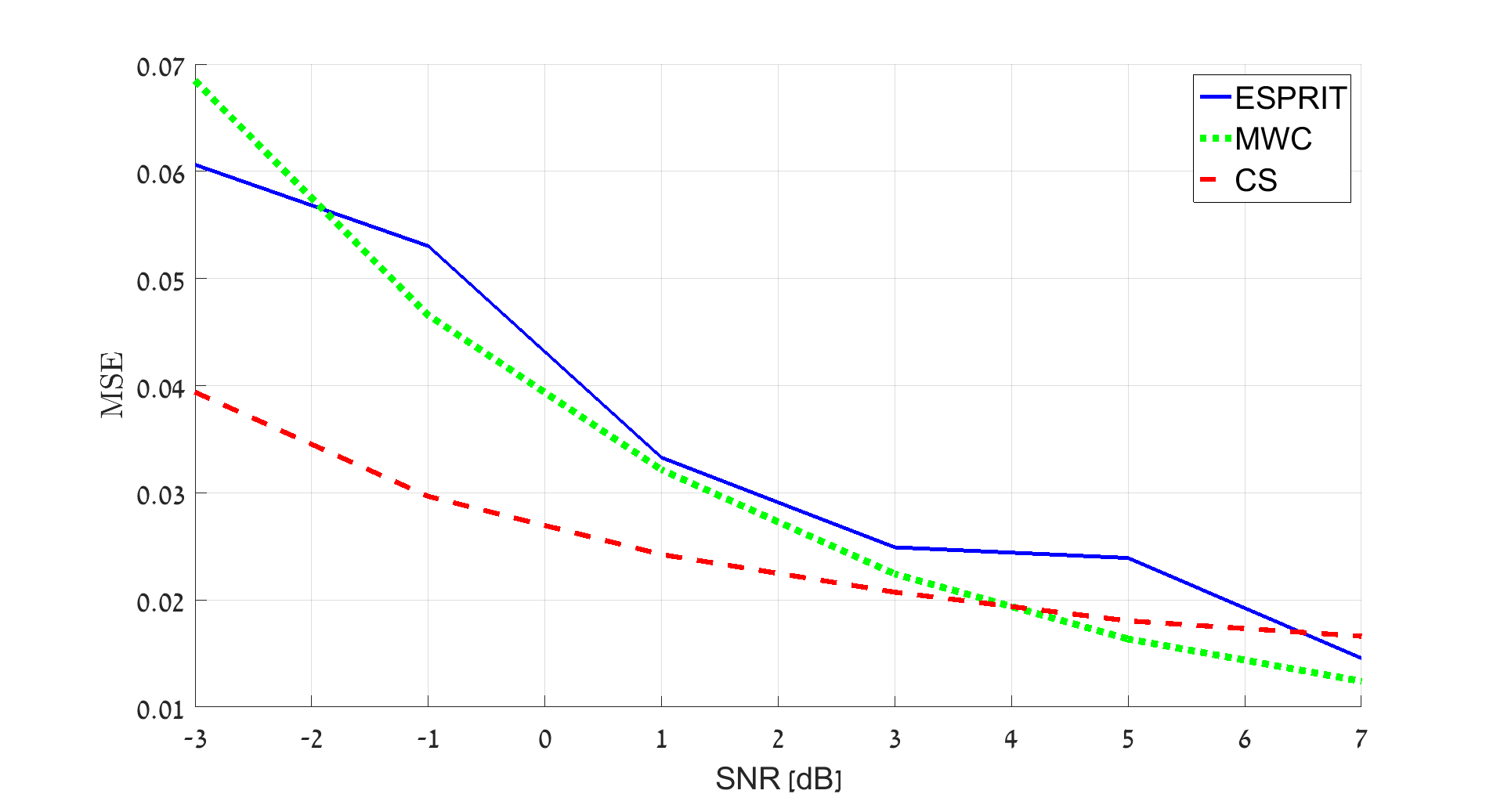}

\protect\caption{Influence of SNR on complex-valued signal reconstruction performance, with $M=3$, $N=10$, $Q=400$. \label{SNR simulation}}
\end{figure}


The simulations demonstrate that our system outperforms the MWC, in particular in low SNR regimes. Besides, in such settings, the reconstruction error of the CS approach is typically lower than that of ESPRIT, which is an analytic method. In the presence of enough samples, originated by increasing the sampling rate $f_s$, the number of sensors $N$ or snapshots $Q$, ESPRIT achieves better results.

\subsection{Carrier Frequency Recovery}

We now consider the case where only the carrier frequencies are recovered,
which can be relevant for various applications such as CR. Here, we
use the following performance measure, $\frac{1}{Mf_{\text{Nyq}}}\sum_{i=1}^{M}\left|f_{i}-\hat{f}_{i}\right|$.
We demonstrate that, for this purpose, a lower sampling rate can be
used. We sample the data at the cut-off frequency of the LPF $f_{s}<f_{p}$, which causes loss of information. Since lower
sampling rate yields fewer samples for a given sensing time, we use
the correlated case or smoothing approach. In the first experiment, we examine
different sampling ratios $\frac{f_{s}}{f_{p}}$. Figure \ref{fs simulation}
demonstrates that even for very low ratios, carrier
frequency recovery yields low error, which decreases as the ratio grows, as expected.
The second simulation, presented in Fig. \ref{fs2 simulation}, shows
the impact of SNR for a fixed sampling ratio $\frac{f_{s}}{f_{p}}=0.2$.

\begin{figure}
\includegraphics[width = 0.5\textwidth]{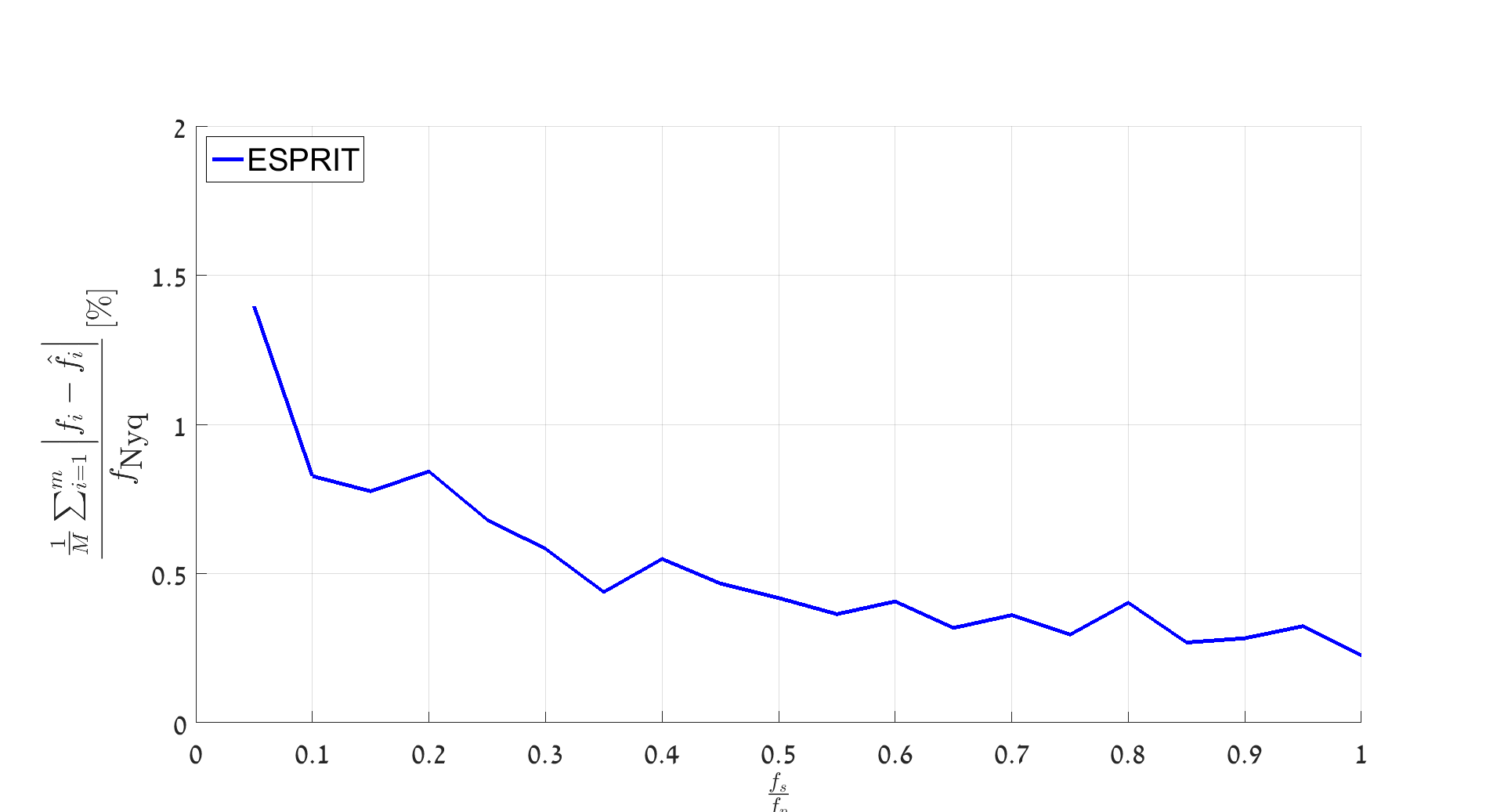}

\protect\caption{Influence of the ratio $\nicefrac{f_{s}}{f_{p}}$ on carrier frequency reconstruction performance, with $M=3$, $N=8$, $Q=400$, $\text{SNR}=10$dB. \label{fs simulation}}
\end{figure}

\begin{figure}
\includegraphics[width = 0.5\textwidth]{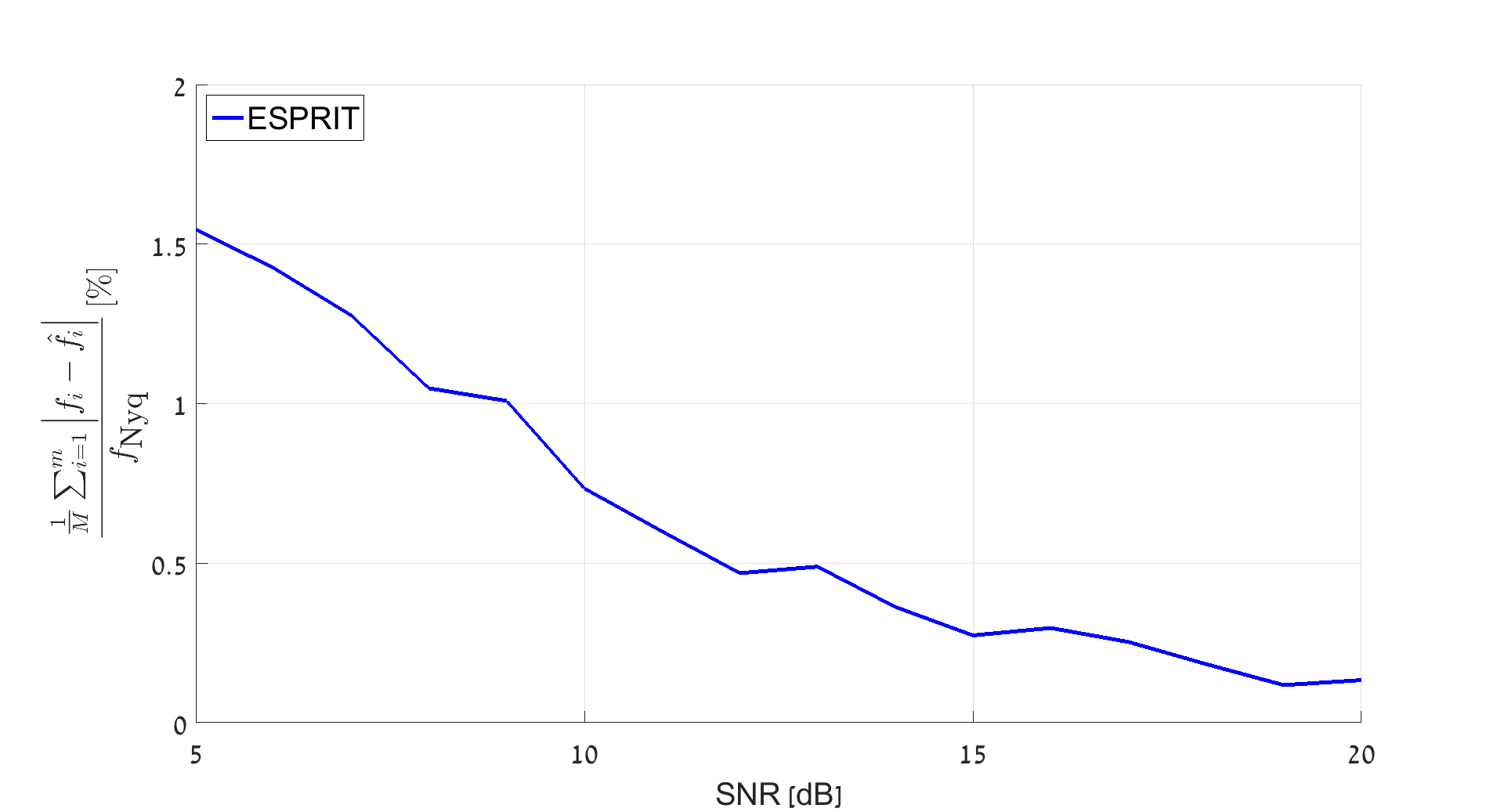}

\protect\caption{Influence of SNR on carrier frequency reconstruction performance, $M=3$, $N=8$, $Q=400$, $\nicefrac{f_{s}}{f_{p}}=0.2$. \label{fs2 simulation}}
\end{figure}

\section{Joint Spectrum Sensing and DOA Recovery}
\label{sec:joint}

We now show how our ULA based system can be expanded to allow for joint recovery of the carrier frequencies and the AOAs.  This is the main advantage of our system with respect to the MWC. We present the compressed carrier and DOA estimation (CaSCADE) system, consisting of an L-shaped array composed of two orthogonal ULAs with an identical sampling scheme. 

Specifically, we consider the problem where the source signals $s_{i}(t), 1 \leq i \leq M$
have both unknown and different carrier frequencies $f_i$ and AOAs $\theta_i$. The main
difference between this scenario and the one that has been discussed
in the previous sections is the additional unknown AOA vector $\boldsymbol{\theta}=\left[\theta_{1}, \theta_{2},\cdots, \theta_{M}\right]^{T}$.
This problem can be treated as a 2D-DOA recovery problem, where two
angles are traditionally recovered. In our case, the second variable
is the signal's carrier frequency instead of an additional angle. The
2D-DOA problem requires both finding the two unknown angles and pairing
them. Previous work \cite{Jgu07, Jgu15} suggests a modification to the
ESPRIT algorithm, that achieves automatic pairing between the two
estimated factors, by simultaneous singular value decomposition (SVD) of
two cross-correlation matrices. 
We further develop this approach, derived in the Nyquist regime, to perform recovery from sub-Nyquist samples. 

\subsection{Signal Model}

In this scenario, for the sake of simplicity, we consider a statistical model.
Let $u\left(t\right)$ and $s_{i}\left(t\right)$ be defined as in
the previous section, with Fourier transforms $U\left(f\right)$ and
$S_{i}\left(f\right)$, accordingly. The signals $s_{i}\left(t\right)$
are considered to be within the $xz$ plane and associated with an AOA
$\theta_{i}$, where $\theta_{i}$ is measured from the positive side
of the $x$ axis. All signals are assumed to be far-field, non coherent,
wide-sense stationary with zero mean and uncorrelated, i.e. $\forall t,\, \mathbb{E}\left[s_{i}\left(t\right)\bar{s}_{j}\left(t\right)\right]=0$
for $i\neq j$, with $\sigma^2_i = \mathbb{E}\left[s_{i}^{2}(t)\right] \neq 0$. Fig. \ref{fig:XZ Plane} illustrates our signal model. To ensure an array structure deprived of ambiguity, we assume that the electronic angles, namely $f_i \cos (\theta_i)$ and $f_i \sin (\theta_i)$, are distinct \cite{multiESPRIT, kikuchi}, namely
\begin{eqnarray}
f_i \cos (\theta_i) \neq f_j \cos (\theta_j), \nonumber \\ 
f_i \sin (\theta_i) \neq f_j \sin(\theta_j),
\label{eq:distinct_angles}
\end{eqnarray}
for $i \neq j$.

\begin{figure}
\begin{centering}
\fbox{%
\begin{minipage}[t]{1\columnwidth}%
\begin{center}
\includegraphics[width = 0.75\textwidth]{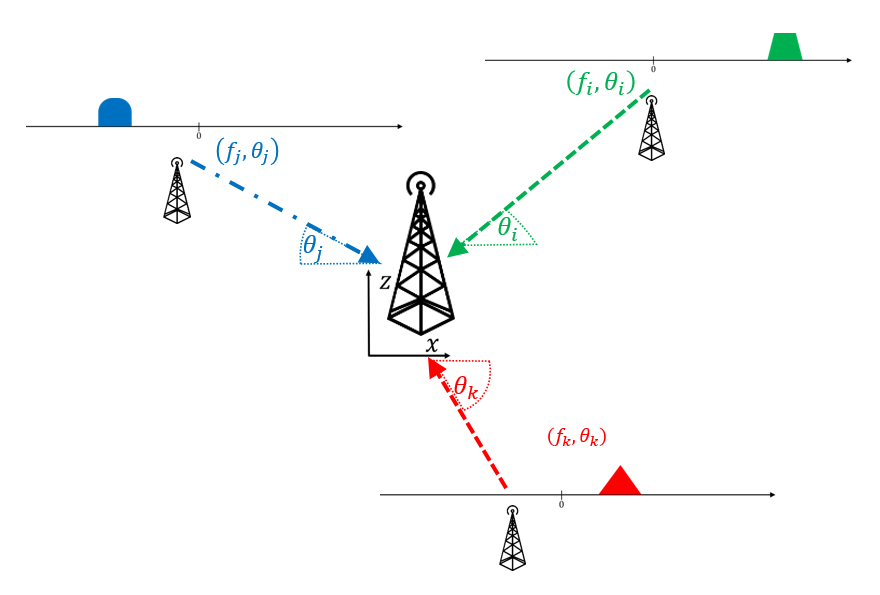}
\par\end{center}%
\end{minipage}}
\par\end{centering}

\protect\caption{Example of $M=3$ source signals in the $xz$ plane. Each transmission is associated with a carrier
frequency $f_{i}$ and AOA $\theta_{i}$.\label{fig:XZ Plane}}
\end{figure}

\begin{definition} The set $\mathcal{M}_{2}$ contains all signals
$u(t)$, such that the support of the Fourier transform
$U(f)$ is contained within a union of $M$ disjoint intervals
in $\mathcal{F}$. Each of the bandwidths does not exceed $B$ and
the transmissions composing $u(t)$ are wide-sense stationary, zero mean and uncorrelated and
have unknown and distinct AOAs $|\theta_{i}|<90^{\circ}$, such that (\ref{eq:distinct_angles}) holds. \end{definition}

In this section, we wish to design a sampling and reconstruction system
which allows for perfect blind signal reconstruction, i.e. recovery
of $\boldsymbol{\theta}, \boldsymbol{f}, \mathbf{s}(t)$, where $\boldsymbol{\theta}$ denotes the AOAs vector defined above and $\boldsymbol{f}$, $\mathbf{s}(t)$ are defined in Section \ref{sec:model}, without any prior knowledge on the carrier frequencies nor the AOAs.

\subsection{CaSCADE System Description \label{sub:ULA description}}
Each transmission $s_{i}\left(t\right)$ impinges on an L-shaped array
with $2N-1$ sensors ($N$ sensors along the $x$ axis and $N$ sensors along the $z$ axis including a common sensor at the origin) in the $xz$ plane with its corresponding AOA
$\theta_{i}$, as shown in Fig.~\ref{fig:L-Shape}. All the sensors
have the same sampling pattern as described in Section \ref{sec:Sampling-Scheme}. In the following sections, we demonstrate that in this case the minimal number of sensors required is $2M$. This leads to a minimal sampling rate of $2MB$ which
is assumed to be less then $f_{\text{Nyq}}$.

\begin{figure}
\begin{centering}
\fbox{%
\begin{minipage}[t]{1\columnwidth}%
\begin{center}
\includegraphics[width = 0.6\textwidth]{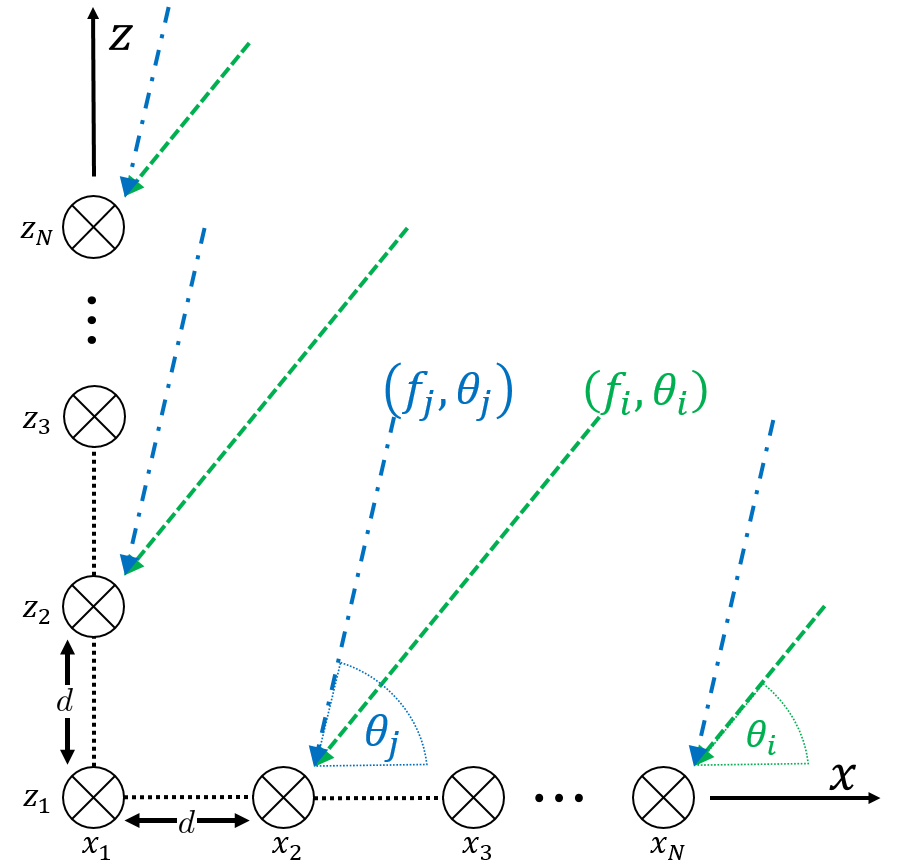}
\par\end{center}%
\end{minipage}}
\par\end{centering}

\protect\caption{CaSCADE system: L-shaped array with $N$ sensors along the $x$ axis and $N$ sensors along the $z$ axis including a common sensor at the origin. \label{fig:L-Shape}}
\end{figure}

By treating the L-shaped array as two orthogonal ULAs, one along the
$x$ axis and the other along the $z$ axis, we form two systems of equations, following the derivations of Section \ref{sec:der1}.
For the ULA along the $x$ axis, we obtain
\begin{equation}
\label{eq:Xeq}
\mathbf{X}(f)=\mathbf{A}_{x} \mathbf{W}(f), \quad f \in \mathcal{F}_s,
\end{equation}
where
\begin{equation}
\mathbf{A}_{x} =
\left[\begin{matrix}e^{j2\pi f_{1}\tau_{1}^{x}(\theta_{1})} & \cdots & e^{j2\pi f_{M}\tau_{1}^{x}(\theta_{M})}\\
\vdots &  & \vdots\\
\\
e^{j2\pi f_{1}\tau_{N}^{x}(\theta_{1})} & \cdots & e^{j2\pi f_{M}\tau_{N}^{x}(\theta_{M})}
\end{matrix}\right].
\end{equation}
Similarly, along the $z$ axis, we get
\begin{equation}
\label{eq:Zeq}
\mathbf{Z}(f)=\mathbf{A}_{z} \mathbf{W}(f), \quad f \in \mathcal{F}_s,
\end{equation}
where $\mathbf{A}_{z}$ is defined accordingly.
Here, $\tau_{n}^{x}\left(\theta\right)=\frac{dn}{c}\cos\left(\theta\right)$, $\tau_{n}^{z}\left(\theta\right)=\frac{dn}{c}\sin\left(\theta\right)$ and the matrices $\mathbf{A}_x$ and $\mathbf{A}_z$ thus depend on both the unknown carrier frequencies $\boldsymbol{f}$ and AOAs $\boldsymbol{\theta}$, namely $\mathbf{A}_x= \mathbf{A}_x\left(\boldsymbol{f},\boldsymbol{\theta}\right)$ and $\mathbf{A}_z= \mathbf{A}_z\left(\boldsymbol{f},\boldsymbol{\theta}\right)$.
In the time domain,
\begin{eqnarray}
\mathbf{x}[k] & = & \mathbf{A}_{x}\mathbf{w}[k], \quad k \in \mathbb{Z} \label{eq:xVsw}\\
\mathbf{z}[k] & = & \mathbf{A}_{z}\mathbf{w}[k], \quad k \in \mathbb{Z}, \label{eq:zVsw}
\end{eqnarray}
where \textbf{$\mathbf{x}[k]$} and $\mathbf{z}[k]$ are the samples for the $x$ and $z$ axis, respectively, and $\mathbf{w}[k]$ is a vector of length $M$ with $i$th element $w_{i}[k]$. In the
following sections, we discuss two possible methods to recover $\boldsymbol{f}$ and $\boldsymbol{\theta}$, present
sufficient conditions to recover the transmissions $\mathbf{s}\left(t\right)$
from $\mathbf{w}[k]$, and provide concrete reconstruction algorithms.

\subsection{Joint ESPRIT Recovery \label{sub:Joint-SVD-ESPRIT}}

We now extend the ESPRIT approach to a 2D setting, in order to jointly recover two parameters, $f_i$ and $\theta_i$, for each transmission. Once these are estimated, the transmissions $s_i(t)$ can be recovered from (\ref{eq:rec_sig}) and (\ref{eq:sVsw}) with the observation matrix $\mathbf{A}= \left[ \begin{array}{l} \mathbf{A}_x \\ \mathbf{A}_z \end{array} \right]$ and the concatenated vector of measurements $\left[ \begin{array}{l} \mathbf{X}(f) \\ \mathbf{Z}(f) \end{array} \right]$.

Consider two sub-arrays of size $N-1$ along each of the $x$ and $z$ axis. The
first sub-array along the $x$ axis consists of sensors $\left\{ 1,...,N-1\right\} $. The second sub-array is
composed of the last $N-1$ sensors along the same axis, i.e. sensors
$\left\{ 2,...,N\right\} $. The sub-arrays along
the $z$ axis are similarly defined. Dropping the time variable $k$ for clarity, we can then write:
\begin{eqnarray}
\mathbf{x}_{1} = \mathbf{A}_{x_{1}}\mathbf{w}, & \quad & \mathbf{x}_{2} =  \mathbf{A}_{x_{2}}\mathbf{w} \nonumber \\
\mathbf{z}_{1} = \mathbf{A}_{z_{1}}\mathbf{w}, & \quad & \mathbf{z}_{2} = \mathbf{A}_{z_{2}}\mathbf{w},
\end{eqnarray}
where $\mathbf{x}_{1}$ and $\mathbf{A}_{x_{1}}$ are the
first $N-1$ rows of $\mathbf{x}$ and $\mathbf{A}_{x}$
respectively and $\mathbf{x}_{2}$ and $\mathbf{A}_{x_{2}}$
are the last $N-1$ rows of $\mathbf{x}$ and $\mathbf{A}_{x}$
respectively. The vectors $\mathbf{z}_1$, $\mathbf{z}_2$ and matrices $\mathbf{A}_{z_1}$, $\mathbf{A}_{z_2}$ are similarly defined.

Each couple of sub-array matrices along the same axis are related as follows:
\begin{eqnarray}
\mathbf{A}_{x_{2}} & = & \mathbf{A}_{x_{1}}\mathbf{\Phi} \nonumber \\
\mathbf{A}_{z_{2}} & = & \mathbf{A}_{z_{1}}\mathbf{\Psi},
\end{eqnarray}
where
\begin{eqnarray}
\label{eq:phipsi}
\mathbf{\Phi} & \triangleq & \mbox{diag}\left[\begin{matrix}e^{j2\pi f_{1}\tau_{1}^{x}(\theta_{1})} & \cdots & e^{j2\pi f_{M}\tau_{1}^{x}(\theta_{M})}\end{matrix}\right] \nonumber \\
\mathbf{\Psi} & \triangleq & \mbox{diag}\left[\begin{matrix}e^{j2\pi f_{1}\tau_{1}^{z}(\theta_{1})} & \cdots & e^{j2\pi f_{M}\tau_{1}^{z}(\theta_{M})}\end{matrix}\right].
\end{eqnarray}
We can see from (\ref{eq:phipsi}) that the carrier frequencies $f_i$ and AOAs $\theta_i$ are embedded in the diagonal matrices $\bf \Phi$ and $\bf \Psi$. Our goal is thus to jointly recover these matrices in order to be able to pair the corresponding elements $f_{i}\tau_{1}^{x}(\theta_{i})$ and $f_{i}\tau_{1}^{z}(\theta_{i})$. We then show how $f_i$ and $\theta_i$ can be estimated from $\mathbf{\Phi}$ and $\mathbf{\Psi}$ for all $1 \leq i \leq M$.

To this end, we apply the ESPRIT framework to cross-correlation matrices between the sub-arrays of both axis. Consider the following correlation matrices:
\begin{eqnarray} \label{eq:cross}
\mathbf{R}_{1} & \triangleq & \mathbb{E}\left[\mathbf{x}_{1}\mathbf{z}_{1}^{H}\right]=\mathbf{A}_{x_{1}}\mathbf{R}_{w}\mathbf{A}_{z_{1}}^{H}, \nonumber \\
\mathbf{R}_{2} & \triangleq & \mathbb{E}\left[\mathbf{x}_{2}\mathbf{z}_{1}^{H}\right]=\mathbf{A}_{x_{2}}\mathbf{R}_{w}\mathbf{A}_{z_{1}}^{H}=\mathbf{A}_{x_{1}}\mathbf{\Phi}\mathbf{R}_{w}\mathbf{A}_{z_{1}}^{H},\nonumber \\
\mathbf{R}_{3} & \triangleq & \mathbb{E}\left[\mathbf{x}_{1}\mathbf{z}_{2}^{H}\right]=\mathbf{A}_{x_{1}}\mathbf{R}_{w}\mathbf{A}_{z_{2}}^{H}=\mathbf{A}_{x_{1}}\mathbf{R}_{w}\mathbf{\Psi}^{H}\mathbf{A}_{z_{1}}^{H}.
\end{eqnarray}
Since the transmissions $s_i(t)$ are assumed to be uncorrelated, $\mathbf{R}_{w}$ is diagonal. In addition, since $\sigma_i^2 \neq 0$, $(\mathbf{R}_{w})_{ii} \triangleq \mathbb{E} \left[w_i^2[k] \right] \neq 0$ and $\mathbf{R}_{w}$ is invertible. Using the fact that $\mathbf{\mathbf{\Psi}}^{H}$ is diagonal as well, we can write
\begin{eqnarray}
\mathbf{R}_{3} & = & \mathbf{A}_{x_{1}}\mathbf{R}_{w}\mathbf{\Psi}^{H}\mathbf{A}_{z_{1}}^{H}=\mathbf{A}_{x_{1}}\mathbf{\Psi}^{H}\mathbf{R}_{w}\mathbf{A}_{z_{1}}^{H}.
\end{eqnarray}
Define the concatenated covariance matrix
\begin{equation}
\mathbf{R}=\left[\begin{matrix}\mathbf{R}_{1}\\
\mathbf{R}_{2}\\
\mathbf{R}_{3}
\end{matrix}\right]=\left[\begin{matrix}\mathbf{A}_{x_{1}}\\
\mathbf{A}_{x_{1}}\mathbf{\Phi}\\
\mathbf{A}_{x_{1}}\mathbf{\Psi}^{H}
\end{matrix}\right]\mathbf{R}_{w}\mathbf{A}_{z_{1}}^{H}.\label{eq:R Matrix}
\end{equation}
The SVD of $\mathbf{R}$ yields
\begin{equation}
\mathbf{R}=\left[\mathbf{U}_{1}\mathbf{U}_{2}\right]\left[\begin{matrix}\boldsymbol{\Lambda} & \boldsymbol{0}\\
\boldsymbol{0} & \boldsymbol{0}
\end{matrix}\right]\mathbf{V}^{H}, \label{eq:R-SVD}
\end{equation}
where $\bf R$ is full column rank, as we show in Lemma \ref{lem:trilinear kruskal rank}.
Then, the columns of the matrix $[\mathbf{U}_{1}\mathbf{U}_{2}]$ are the left singular vectors of $\mathbf{R}$, where $\mathbf{U}_{1}$ contains the vectors corresponding to the first $M$ singular values, $\mathbf{\Lambda}$ is a $M\times M$ diagonal matrix with the $M$ non zero singular values of $\mathbf{R}$, and $\mathbf{V}$ contains the right singular vectors of $\mathbf{R}$.

We now derive sufficient conditions for perfect recovery of $\bf \Phi$ and $\bf \Psi$, up to some joint permutation, from $\mathbf{U}_1$. We then show how $\bf \Phi$ and $\bf \Psi$, and as a consequence $\boldsymbol{f}$ and $\boldsymbol{\theta}$, can be recovered from $\mathbf{U}_1$. First, Lemma \ref{lem:trilinear kruskal rank} provides sufficient conditions so that there exists an invertible $M \times M$ matrix $\bf T$ such that
\begin{equation}
\mathbf{U}_{1}=\left[\begin{matrix}\mathbf{U}_{11}\\
\mathbf{U}_{12}\\
\mathbf{U}_{13}
\end{matrix}\right]=\left[\begin{matrix}\mathbf{A}_{x_{1}}\\
\mathbf{A}_{x_{1}}\mathbf{\Phi}\\
\mathbf{A}_{x_{1}}\mathbf{\Psi}^{H}
\end{matrix}\right]\mathbf{T}, \label{eq:Signal Sub Space}
\end{equation}
where $\mathbf{U}_{1i}$ are $(N-1)\times M$ matrices.

\begin{lem}
\label{lem:trilinear kruskal rank} Let $u\left(t\right)$ be an arbitrary
signal within $\mathcal{M}_{2}$ and consider an L-shaped ULA with $N$ sensors and distance $d$ between two adjacent sensors. If:
\begin{itemize}
\item (c1) $d \leq \frac{c}{f_{\text{Nyq}}}$
\item (c2) $N>M$
\end{itemize}
then (\ref{eq:Signal Sub Space}) holds.\end{lem}

\begin{IEEEproof}
We begin by showing that under conditions (c1)-(c2), $\bf R$ is full rank. From (\ref{eq:distinct_angles}), both the $(N-1) \times M$ matrices $\mathbf{A}_{x_1}$ and $\mathbf{A}_{z_1}$ are Vandermonde with distinct columns, and are thus full column rank. The $M \times N$ matrix $\mathbf{R}_{w}$ is diagonal and invertible. It follows that $\mathbf{R}_{1}$ and $\mathbf{R}$ are full column rank.

The SVD decomposition of $\mathbf{R}$ yields (\ref{eq:R-SVD}). 
In particular, it holds that $\mathbf{R}^{H}\mathbf{U}_{2}=\boldsymbol{0}$, where $\bf R$ is of size $3(N-1) \times N-1$ with $\text{rank}(\mathbf{R})=M$ and the $3(N-1) \times (3(N-1)-M)$ matrix $\mathbf{U}_2$ is in the null space of $\bf R$. That is
\[
\mathbf{A}_{z_{1}}\mathbf{R}_{w} \mathbf{B} \mathbf{U}_{2}=\boldsymbol{0},
\]
where 
\[
\mathbf{B}=\left[\begin{matrix}\mathbf{A}^H_{x_{1}} & \mathbf{\Phi}^{H}\mathbf{A}_{x_{1}}^{H} &
\mathbf{\Psi}^H\mathbf{A}_{x_{1}}^{H}
\end{matrix}\right].
\]
Since $\mathbf{A}_{z_{1}} \mathbf{R}_{w}$ is full rank, it follows that $\mathbf{B} \mathbf{U}_{2}=\boldsymbol{0}$. Besides $\text{rank}(\mathbf{B}) = \text{rank}(\mathbf{U}_1)$; this implies that $\mbox{span}\left(\mathbf{B}^H\right)=\mbox{span}\left(\mathbf{U}_{1}\right)$. Therefore, there exists a $M\times M$ invertible matrix $\mathbf{T}$
such that (\ref{eq:Signal Sub Space}) holds.
\end{IEEEproof}

If the conditions of Lemma \ref{lem:trilinear kruskal rank} hold, then we can write
\begin{eqnarray}
\mathbf{A}_{x_{1}} & = & \mathbf{U}_{11}\mathbf{T}^{-1} \nonumber \\
\mathbf{U}_{12} & = & \mathbf{A}_{x_{1}}\mathbf{\Phi}\mathbf{T}=\mathbf{U}_{11}\mathbf{T}^{-1}\boldsymbol{\Phi T}\nonumber \\
\mathbf{U}_{13} & = & \mathbf{A}_{x_{1}}\mathbf{\Psi}^{H}\mathbf{T}=\mathbf{U}_{11}\mathbf{T}^{-1}\mathbf{\Psi}^{H}\mathbf{T}. \nonumber 
\end{eqnarray}
Besides, since $\mathbf{U}_{1i}$ is of size $(N-1) \times M$, with $N>M$, the number of its rows is greater or equal to the number of its columns. In addition, from (\ref{eq:Signal Sub Space}), $\mathbf{U}_{11}=\mathbf{A}_{x_1}\mathbf{T}$, where $\bf T$ is invertible and $\text{rank}(\mathbf{A}_{x_1})=M$. Therefore, it holds that $\text{rank}(\mathbf{U}_{11})=M$ and $\mathbf{U}_{11}^{\dagger} \mathbf{U}_{11}=\mathbf{I}$.

We can then derive a relation between $\bf \Phi$, $\bf \Psi$ and the blocks that compose $\mathbf{U}_1$ as
\begin{eqnarray}
\mathbf{U}_{11}^{\dagger}\mathbf{U}_{12} & = & \mathbf{T}^{-1}\boldsymbol{\Phi T}\label{eq:Same Permutation} \nonumber \\
\mathbf{U}_{11}^{\dagger}\mathbf{U}_{13} & = & \mathbf{T}^{-1}\mathbf{\Psi}^{H}\mathbf{T},
\end{eqnarray}
where the matrix $\mathbf{T}$ is identical in both equations. We can now obtain $\mathbf{\Phi}$ and $\mathbf{T}$ using an eigenvalue decomposition up to permutation. Denote by $\bf \hat{\Phi}$ and $\bf \hat{T}$ the obtained matrices. Once these are recovered, we compute $\bf \hat{\Psi}$ with the same permutation, as
\begin{equation}
\hat{\mathbf{\Psi}}^H = \hat{\mathbf{T}}\left(\mathbf{U}_{11}^{\dagger}\mathbf{U}_{13}\right)\hat{\mathbf{T}}^{-1}.
\end{equation}
Since the electronic angles $f_i \cos (\theta_i)$ and $f_i \sin (\theta_i)$ are distinct, the eigenvalues of $\hat{\mathbf{\Phi}}$ and $\hat{\mathbf{\Psi}}$ are distinct as well and if follows that both matrices have the same permutation. We thus obtain proper pairing between the diagonal elements, and the AOAs $\theta_i$ and carrier frequencies $f_i$ are given by
\begin{eqnarray}
\theta_{i} & = & \tan^{-1}\left(\frac{\angle\Psi_{ii}}{\angle\Phi_{ii}}\right) \nonumber \\
f_{i} & = & \frac{\angle\Phi_{ii}}{2\pi\frac{d}{c}\cos\left(\theta_{i}\right)}. \label{eq:DOA solution}
\end{eqnarray}
Algorithm \ref{algo:2desprit} summarizes the main steps of the joint 2D ESPRIT described above. In the algorithm description we exploit the 4 cross-correlation matrices between the sub-arrays instead of only 3 as defined in (\ref{eq:cross}) to increase robustness to noise. Here, we assume perfect knowledge of $\mathbf{R}$. In practice, it can be estimated as shown in Algorithm \ref{algo:2desprit}.

\begin{algorithm} 
\textbf{\uline{Input:}}\textbf{ }
\begin{itemize}
\item $Q$ snapshots of the sensors measurements $\mathbf{x}$ along the $x$ axis
\item $Q$ snapshots of the sensors measurements $\mathbf{z}$ along the $z$ axis
\end{itemize}
\textbf{\uline{Output:}}
\begin{itemize}
\item $\hat{\boldsymbol{f}}$ - estimated carriers frequencies
\item $\hat{\boldsymbol{\theta}} $ - estimated the AOA
\end{itemize}
\textbf{\uline{Algorithm:}}
\begin{enumerate}
\item Define $\mathbf{x}_{1}$ and $\mathbf{x}_{2}$ as the first and
last $N-1$ rows of $\mathbf{x}$\\
Define $\mathbf{z}_{1}$ and $\mathbf{z}_{2}$ as the first and
last $N-1$ rows of $\mathbf{z}$
\item Estimate the cross covariance matrices:

\begin{enumerate}
\item $\mathbf{R}_{1}=\sum_{k=1}^Q \mathbf{x}_{1}[k]\mathbf{z}_{1}^H[k]$
\item $\mathbf{R}_{2}=\sum_{k=1}^Q \mathbf{x}_{2}[k]\mathbf{z}_{1}^H[k]$
\item $\mathbf{R}_{3}=\sum_{k=1}^Q \mathbf{x}_{1}[k]\mathbf{z}_{2}^H[k]$
\item $\mathbf{R}_{4}=\sum_{k=1}^Q \mathbf{x}_{2}[k]\mathbf{z}_{2}^H[k]$
\end{enumerate}
\item Decompose $\mathbf{R}=\left[\begin{matrix}\mathbf{R}_{1}^T &
\mathbf{R}_{2}^T &
\mathbf{R}_{3}^T &
\mathbf{R}_{4}^T
\end{matrix}\right]^T$ using SVD: $\mathbf{U,S,V}=\text{svd}(\mathbf{R})$
\item Set $\mathbf{U}_1$ to be the ${(4N-4)\times M}$ matrix that contains the $M$ left eigenvectors corresponding to the highest singular values of $\mathbf{R}$
\item Define:

\begin{enumerate}
\item $\mathbf{U}_{11}$ as the first $N-1$ rows of $\mathbf{U}$
\item $\mathbf{U}_{12}$ as the next $N-1$ rows of $\mathbf{U}$
\item Same for $\mathbf{U}_{13}$, $\mathbf{U}_{14}$
\end{enumerate}
\item Compute:

\begin{enumerate}
\item $\mathbf{V}_{1}=\mathbf{U}_{11}^{\dagger}\mathbf{U}_{12}$
\item $\mathbf{V}_{2}=\mathbf{U}_{11}^{\dagger}\mathbf{U}_{13}$
\item $\mathbf{V}_{3}=\mathbf{U}_{11}^{\dagger}\mathbf{U}_{14}$
\end{enumerate}
\item Perform an eigenvalue decomposition of $\left(\mathbf{V}_{1}+\mathbf{V}_{2}+\mathbf{V}_{3}\right)=\mathbf{T}\boldsymbol{\Lambda}\mathbf{T}^{-1}$, where $\boldsymbol{\Lambda}$ is a diagonal matrix

\item Compute $\mathbf{\hat{\Phi}}=\mathbf{T}^{-1}\mathbf{V}_{1}\mathbf{T}$
and $\mathbf{\hat{\Psi}}=\left(\mathbf{T}^{-1}\mathbf{V}_{2}\mathbf{T}\right)^{H}$
\item Compute the carrier frequencies and AOAs using (\ref{eq:DOA solution})
\end{enumerate}
\protect\caption{Joint ESPRIT}

\label{algo:2desprit}
\end{algorithm}

Finally, Theorem \ref{thm:DOA unique solution} summarizes sufficient conditions for perfect blind reconstruction of $\boldsymbol{f}$ and $\boldsymbol{\theta}$ from the low rate samples $\mathbf{x}[k]$ and $\mathbf{z}[k]$.

\begin{thm}
\label{thm:DOA unique solution}Let $u(t)$ be an arbitrary
signal within $\mathcal{M}_{2}$ and consider an L-shaped ULA with $2N-1$ sensors, such that there are $N$ sensors along each axis with a common sensor at the origin, and the distance between two adjacent sensors is denoted by $d$. If:
\begin{itemize}
\item (c1) $d<\frac{c}{f_{\text{Nyq}}}$
\item (c2) $N>M$,
\end{itemize}
then (\ref{eq:Xeq})-(\ref{eq:Zeq}) has a unique solution $\left(\boldsymbol{f},\boldsymbol{\theta},\mathbf{w}\right)$.\end{thm}
\begin{IEEEproof}
From Lemma \ref{lem:trilinear kruskal rank}, it follows that, under conditions (c1)-(c2), $\mathbf{U}_{11}$ is full column rank and thus left invertible. Therefore, $\mathbf{\Phi}$ and $\mathbf{\Psi}$ can be uniquely derived from (\ref{eq:Same Permutation}), with the same permutation $\mathbf{T}$ for both matrices. This follows from the assumption that the electronic angles and, as a consequence the eigenvalues of $\mathbf{\Phi}$ and $\mathbf{\Psi}$, are distinct. Condition (c1) implies
that both $2\pi\hat{f}_{i}\frac{d}{c}\cos\left(\hat{\theta}_{i}\right)\in(-\pi,\pi]$
and $2\pi\hat{f}_{i}\frac{d}{c}\sin\left(\hat{\theta}_{i}\right)\in(-\pi,\pi]$
namely $\angle\Psi_{i,i}$ and $\angle\Phi_{i,i}$, for $1 \leq i \leq M$, are unique, and it follows
that $\boldsymbol{f},\boldsymbol{\theta}$ are unique
as well and are given by (\ref{eq:DOA solution}).
\end{IEEEproof}

In addition, if conditions (c1)-(c2) from Theorem \ref{prop:signals-recovery:-given} hold, then $s_i(t)$ is uniquely recovered from $\mathbf{x}[k]$ and $\mathbf{z}[k]$ from (\ref{eq:rec_sig}) and (\ref{eq:sVsw}) with the observation matrix $\mathbf{A}= \left[ \begin{array}{l} \mathbf{A}_x \\ \mathbf{A}_z \end{array} \right]$.

From Theorem \ref{thm:DOA unique solution}, the minimal necessary number of sensors in each axis, including a common sensor at the origin, to allow perfect blind reconstruction is $N\geq M+1$, leading to a total number of sensors $2N-1 \geq 2M +1$ sensors in each axis, including a common sensor at the origin. 
In addition, for perfect reconstruction we require $f_{s}\geq B$ as in Theorem \ref{prop:signals-recovery:-given}. Thus, the minimal sampling rate is bounded by $(2M+1)B$.

\subsection{CS Approach \label{sub:Compressed-Sensing-Approach}}
In this section, we derive a second joint carrier frequency and AOA recovery approach based on CS methods. Denote 
\begin{equation}
\mathbf{v}[k] = \left[ \begin{matrix} \mathbf{x}[k] \\ \mathbf{z}[k] \end{matrix} \right]
\end{equation}
that stacks the samples from sensors of both axis. Consider the correlation matrix
\begin{equation}
\mathbf{R}=\mathbb{E}\left[\mathbf{v}[k] \mathbf{v}^{H}[k]\right]=\mathbf{A}\mathbf{R}_{w}\mathbf{A}^{H},\label{eq:rvsp}
\end{equation}
where $\mathbf{A} = \left[ \begin{matrix} \mathbf{A}_x \\ \mathbf{A}_z \end{matrix} \right]$.
In the following, we assume perfect knowledge of $\mathbf{R}$.
In practice, it can be estimated as
\begin{equation}
\mathbf{R}=\sum_{k=1}^{Q}\mathbf{v}[k]\mathbf{v}[k]^{H},\label{eq:r_est}
\end{equation}
where $Q$ is the number of snapshots.

Denote $\alpha_{i}=f_{i}\cos\theta_{i}$ and $\beta_{i}=f_{i}\sin\theta_{i}$
and suppose that $\alpha_{i}$ and $\beta_{i}$ lie on a grid $\{\delta l\}_{l=-L_0}^{L_0}$,
with $L_0=\frac{f_{\text{Nyq}}}{2\delta}$. Here, $\delta$ is a parameter
of the recovery algorithm that defines the grid resolution. With high probability, the discretization conserves the unambiguous property, namely $\text{spark}(\mathbf{G})=N+1$. Formulating concrete conditions to ensure the lack of ambiguity is very involved and thus this property is traditionally assumed without justification \cite{multiESPRIT}.

Denote $L=2L_0+1$. We can then write
\begin{equation}
\mathbf{R}=\mathbf{G}\mathbf{R}_w^g\mathbf{G}^{H},\label{eq:rvsp2}
\end{equation}
where ${\bf G}$ is a $(2N-1)\times L^2$ matrix
with $(n,l)$th element $G_{nl}=e^{j2\pi\frac{dn}{c}\alpha_{l_1}}$, for $0 \leq n \leq N-1$ and $G_{nl}=e^{j2\pi\frac{d (n-N+1)}{c} \beta_{l_2} }$, for $N \leq n \leq 2N-1$. Here, $l_1=(l \mod L)-L_0$ and $l_2=\lfloor \frac{l}{L} \rfloor -L_0$. The nonzero elements of the $L^2 \times L^2$ matrix $\mathbf{R}_w^g$ are the $M$ diagonal elements of $\mathbf{R}_{w}$
at the $M$ indices corresponding to $\{\alpha_{i},\beta_{i}\}$.
Since $\mathbf{R}_w^g$ is diagonal, the observation model (\ref{eq:rvsp2}) can be equivalently written in vector form as
\begin{equation}
\label{eq:vecCS2}
\text{vec}(\mathbf{R})= \left( \bar{\mathbf{G}} \odot \mathbf{G} \right) \mathbf{r}_w^g.
\end{equation}
Here $\text{vec}(\mathbf{R})$ is a column vector that vectorizes the matrix $\bf R$ by stacking its columns, $\mathbf{r}_w^g$ is the $L^2 \times 1$ vector that contains the diagonal of $\mathbf{R}_w^g$ and $\odot$ denotes the Khatri-Rao product. 
The goal is thus to recover the $M$-sparse vector $\mathbf{r}_w^g$ from the $(2N-1)^2$
measurement vector $\text{vec}(\mathbf{R})$. 

The following theorem derives a necessary condition on the minimal
number of sensors $2N-1$ for perfect recovery of $\alpha_{i},\beta_{i}$,
$i\in\{1,\dots,M\}$ in a noiseless environment.
\begin{thm}
Let $u(t)$ be an arbitrary signal within $\mathcal{M}_{2}$ and consider an L-shaped ULA with $2N-1$ sensors, such that there are $N$ sensors along each axis with a common sensor at the origin, and the distance between two adjacent sensors is denoted by $d$. If:
\begin{itemize}
\item (c1) $d<\frac{c}{f_{\text{Nyq}}}$
\item (c2) $N>M$
\item (c3) $\text{spark}(\mathbf{G})=N+1$,
\end{itemize}
then (\ref{eq:vecCS2}) has a unique $M$-sparse solution $\mathbf{r}_w^g$.
\end{thm}
\begin{proof}
In order to recover the $M$-sparse vector $\mathbf{r}_w^g$ from $\text{vec}(\mathbf{R})$, we require \cite{CSBook, SamplingBook}
\begin{equation}
\label{eq:spark_kr}
\text{spark} \left( \bar{\mathbf{G}} \odot \mathbf{G} \right) > 2M.
\end{equation}
From \cite{KR_spark}, it holds that
\begin{equation}
\text{spark} \left( \bar{\mathbf{G}} \odot \mathbf{G} \right) \geq \min \{ 2(\text{spark} (\mathbf{G}) -1), L^2+1 \}.
\end{equation}
Combining (c2) and (c3), we have 
\begin{equation}
2(\text{spark} (\mathbf{G}) -1) > 2M.
\end{equation}
Finally, since $L^2 \gg N$, (\ref{eq:spark_kr}) holds.
\end{proof}

To recover the sparse vector $\mathbf{r}_w^g$, we can use any CS recovery algorithm such as orthogonal matching pursuit (OMP) \cite{CSBook}.
Once the indices $\alpha_{i},\beta_{i}$, $i\in\{1,\dots,M\}$
are recovered, the corresponding $f_{i}$ and $\theta_{i}$ are given
by
\begin{eqnarray}
\hat{\theta}_{i} & = & \tan^{-1}\left(\frac{\beta_{i}}{\alpha_{i}}\right),\nonumber \\
\hat{f}_{i} & = & \frac{\alpha_{i}}{\cos(\theta_{i})}.
\end{eqnarray}

\subsection{Numerical Results}

In this section, we demonstrate the effect of different system parameters
on the reconstruction performance. Consider a complex-valued signal
$u\left(t\right)$ from $\mathcal{M}_{2}$, which is the sum of $M=3$
narrowband source signals $s_{i}\left(t\right)$,
each of width $B=50$Mhz and with $f_{\text{Nyq}}=10$Ghz. The carrier
frequencies $f_{i}$ are drawn uniformly at random from $[-\frac{f_{\text{Nyq}}-B}{2},\frac{f_{\text{Nyq}}-B}{2}]$,
and the AOAs $\theta_{i}$ are drawn uniformly at random from $[-85^{\circ},85^{\circ}]$.
The L-shaped array is composed of $2N-1$ sensors; $N$ along each axis with a common sensor at the origin. The received signal at
each sensor is corrupted with AWGN. The mixing and sampling rates
are set to $f_{s}=f_{p}=1.4B$.

In this section, we compare 3
reconstruction methods: 1) Joint ESPRIT summarized in Algorithm \ref{algo:2desprit}, 2) CS approach presented in Section \ref{sub:Compressed-Sensing-Approach}, 3) PARAFAC analysis \cite{Hars94} based approach, as presented in \cite{Stein2015}. PARAFAC \cite{Hars94} extends the bilinear model of factor analysis to a trilinear model using the alternating least squares (ALS) method. In \cite{Stein2015}, PARAFAC is used to decompose the cross correlations matrices defined in (\ref{eq:cross}) into three matrices, isolating $\bf \Phi$ and $\bf \Psi$. To apply the PARAFAC algorithm we use the COMFAC MATLAB function implemented
by \cite{Sidi98}.

In these simulations, we focus on the recovery of the carrier frequencies
$f_{i}$ and AOAs $\theta_{i}$. Once these are recovered, full signal
reconstruction can be performed as shown in the first part of this
work (see (\ref{eq:rec_sig}) and (\ref{eq:sVsw})). The reconstruction performance is measured by the following
criteria: $\frac{1}{Mf_{\text{Nyq}}}\sum_{i=1}^{m}\left|f_{i}-\hat{f}_{i}\right|$
for the frequencies, and $\frac{1}{M180^{\circ}}\sum_{i=1}^{m}\left|\theta_{i}-\hat{\theta}_{i}\right|$
for the AOA.

The first simulation examines the recovery performance with respect
to the number of sensors $2N-1$. Figure \ref{DOA sensors} presents the carrier frequency
and AOA reconstruction performance for different values of the number
of sensors, which affects both the noise averaging and the total amount of samples
available.

\begin{figure}[H]
\includegraphics[width = 0.5\textwidth]{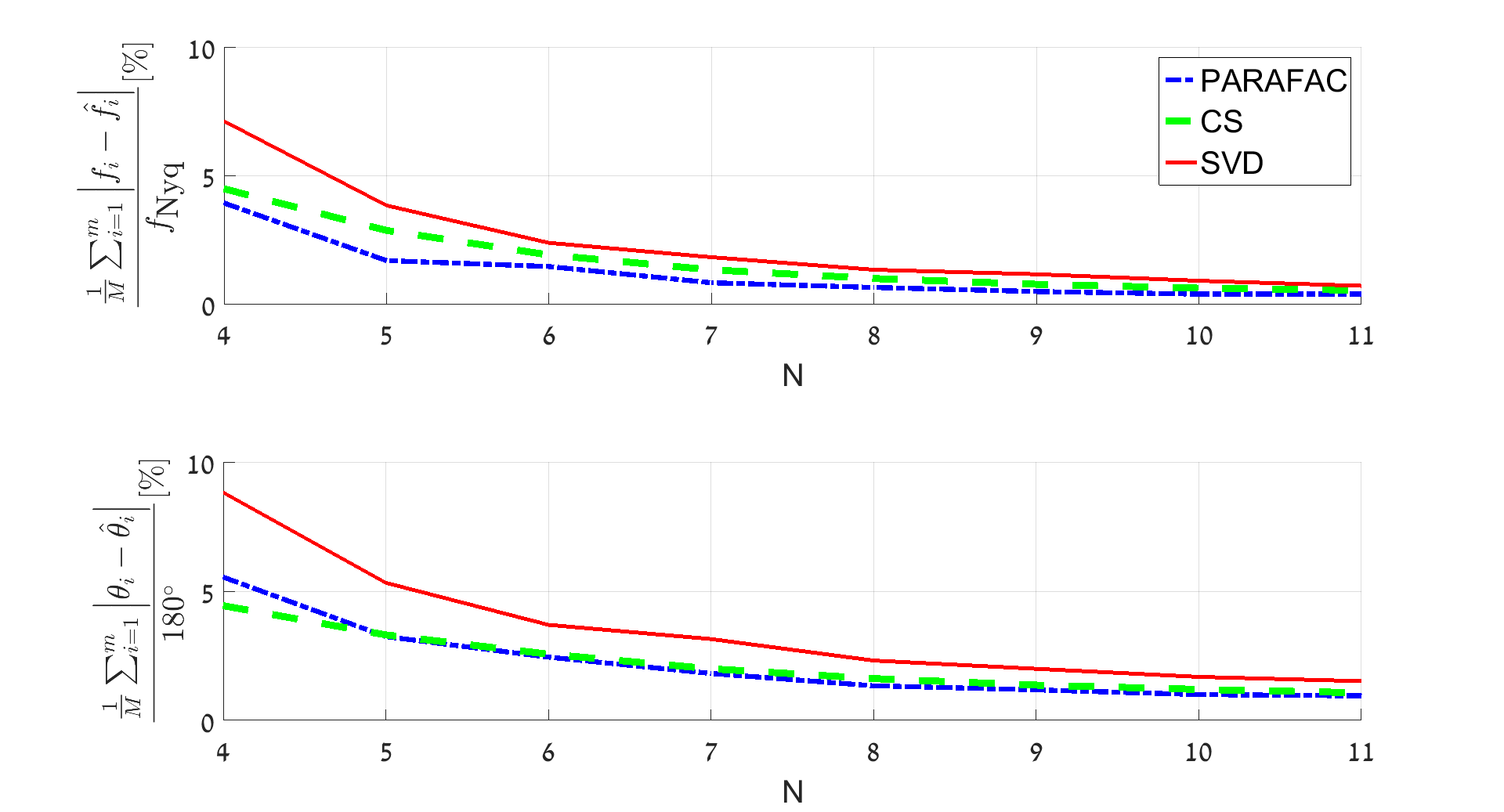}
\protect\caption{Influence of the number of sensors $2N-1$ on the carrier frequency and AOA reconstruction performance, with $M=3$, $Q=300$,
$\text{SNR}=10$dB. \label{DOA sensors}}
\end{figure}

The second simulation, presented in Fig. \ref{DOA SNR}, illustrates
the impact of SNR on the recovery performance.

\begin{figure}[H]
\includegraphics[width = 0.5\textwidth]{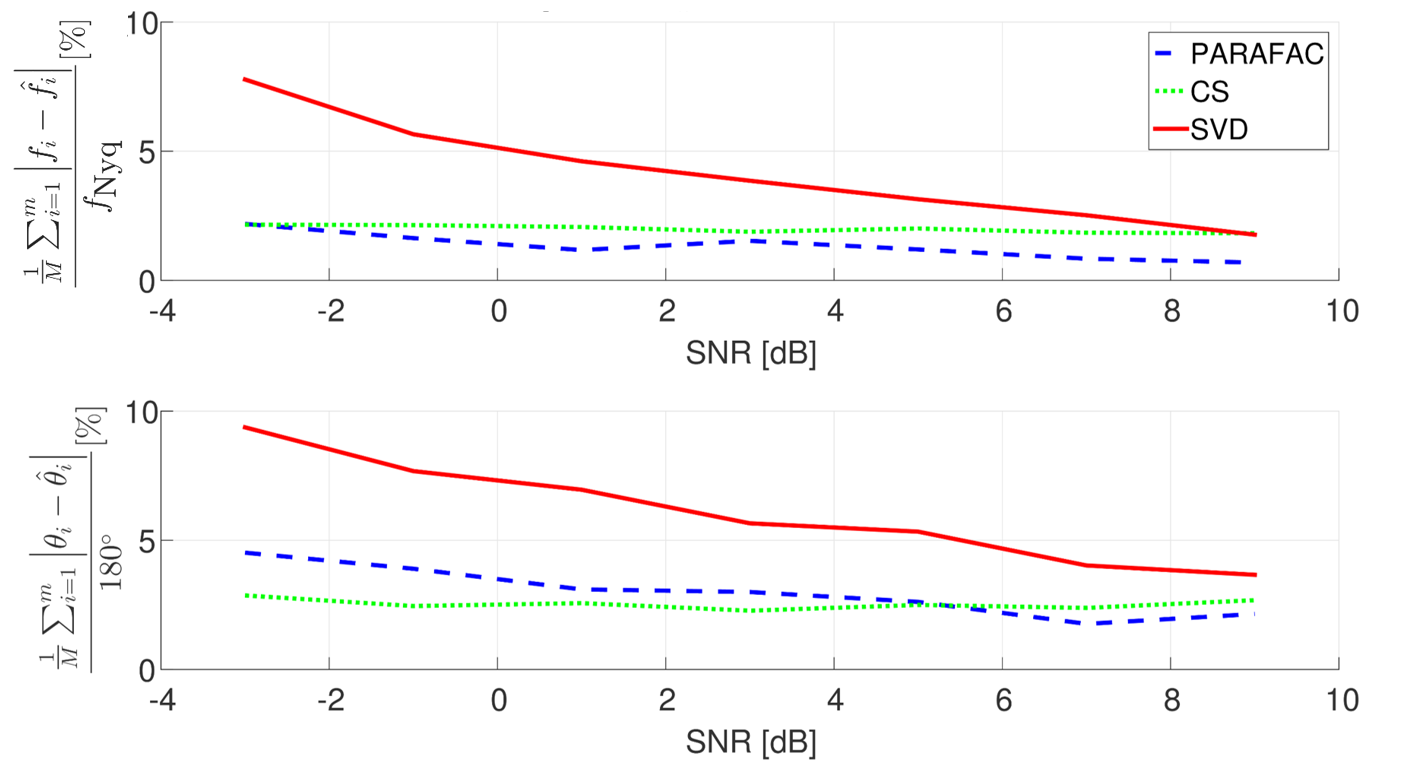}
\protect\caption{Influence of SNR on the carrier frequency and AOA reconstruction performance, with $M=3$, $2N-1=11$, $Q=300$. \label{DOA SNR}}
\end{figure}

\section{Conclusion}
In this paper, we considered two scenarios: spectrum sensing and joint spectrum sensing and DOA of multiband signals from sub-Nyquist samples. For the first scenario, we proposed a receiver composed of a ULA, where each sensor contains an analog front-end equivalent to one channel of the MWC. This system constitutes an alternative sub-Nyquist sampling scheme that outperforms the MWC in terms of performance in low SNRs or implementation complexity. For the joint spectrum sensing and DOA scenario, we extend our ULA configuration and present the CasCADE system, an L-shaped array composed of two ULAs with the same sampling scheme as above. In both cases, we derive sufficient conditions for the recovery of the transmissions carrier frequencies and AOAs, if relevant. We showed that the minimal number of sensors for the first scenario is twice the number of transmissions, namely $2M$, in the worst case and $M+1$ with high probability, whereas in the second scenario, it is $2M+1$ in the average case. Last, we provided two reconstruction schemes for both scenarios: one based on the analytic method ESPRIT and the second based on CS techniques. Simulations demonstrated the performance of the above algorithms in comparison with existing methods.

\bibliographystyle{IEEEtran}
\bibliography{IEEEabrv,IEEEfull}

\end{document}